\documentclass{article} 

\usepackage[noend,linesnumbered,ruled]{algorithm2e}
\newcommand{\unwind}{\sim}
\newcommand{\ocp}{$\text{OC}^\text{P}$\xspace}
\newcommand{\scp}{$\text{SC}^\text{P}$\xspace}
\newcommand{\lrp}{$\text{LR}^\text{P}$\xspace}
\newcommand{\ocip}{$\text{OC}^\text{IP}$}
\newcommand{\scip}{$\text{SC}^\text{IP}$}
\newcommand{\lrip}{$\text{LR}^\text{IP}$}
\newcommand{\octa}{$\text{OC}^\text{TA}$}
\newcommand{\scta}{$\text{SC}^\text{TA}$}
\newcommand{\lrta}{$\text{LR}^\text{TA}$}
\newcommand{\makeset}{\textsc{Make{-}Set}}
\newcommand{\union}{\textsc{Union}}
\newcommand{\findset}{\textsc{Find}}
\newcommand{\eqifind}[1]{\sim^{\findset}_{#1}}

\newcommand{\eqiscall}[1]{\sim^{\text{SC}}_{#1}}
\newcommand{\eqiobs}[1]{\sim^{\obs_{#1}}}

\newcommand{\ftview}{{\tt ftview}}
\newcommand{\lpre}{{\tt lpre}}
\newcommand{\final}[1]{{#1}^{\textsf{f}}}
\newcommand{\powerset}[1]{\mathcal{P}(#1)}
\newcommand{\convert}{\text{convert}}
\newcommand{\convertback}{\text{convertback}}

\newcommand{\qed}{$\Box$\\}
\newenvironment{proof}{\noindent {\bf Proof:} }{\hfill\qed}

\usepackage{paralist}

\usepackage{amssymb,color}
\usepackage[cmex10]{amsmath}
\usepackage{paralist}
\usepackage{graphicx}

\newcommand{\commentout}[1]{}

\newcommand{\nintrel}{\rightarrowtail}
\newcommand{\nintrelimage}[1]{{#1}^\nintrel}

\newcommand{\step}{{\tt step}}
\newcommand{\Actions}{A}
\newcommand{\Dom}{D}
\newcommand{\States}{S} 
\newcommand{\Observations}{O}
\newcommand{\dom}{{\tt dom}} 
\newcommand{\obs}{{\tt obs}}

\newcommand{\view}{{\tt view}}
\newcommand{\tview}{{\tt tview}}
\newcommand{\sources}{{\tt sources}}
\newcommand{\ipurge}{{\tt ipurge}}
\newcommand{\purge}{{\tt purge}}
\newcommand{\ta}{{\tt ta}} 
\newcommand{\pt}{{\tt to}}
\newcommand{\ito}{{\tt ito}}
\newcommand{\Alphabet}[1]{\textnormal{alph}(#1)}

\newcommand{\ipsec}{IP-se\-cure}
\newcommand{\psec}{P-se\-cure}
\newcommand{\tosec}{TO-se\-cure}
\newcommand{\itosec}{ITO-secure}
\newcommand{\tasec}{TA-se\-cure}
\newcommand{\spurge}[2]{#1{\restrict}#2}
\newcommand{\restrict}{\upharpoonleft}

\newcommand{\ipsecty}{IP-security}
\newcommand{\psecty}{P-security}
\newcommand{\tosecty}{TO-security}
\newcommand{\tasecty}{TA-security}
\newcommand{\itosecty}{ITO-security}

\newcommand{\Swap}{\leftrightarrow^\textsl{swap}}
\newcommand{\RefSwap}{\stackrel = \leftrightarrow^\textsl{swap}}
\newcommand{\Symrefswap}{\equiv^\textsl{oi}}
\newcommand{\RedIrr}{\rightarrow^\textsl{irr}}
\newcommand{\SymRedIrr}{\leftrightarrow^\textsl{irr}}
\newcommand{\SymrefRedIrr}{\stackrel = \leftrightarrow^\textsl{irr}}

\newcommand{\acat}{\circ} 
\newcommand{\pow}[1]{{\cal P}(#1)}
\newcommand{\hset}{{\cal T}}

\newtheorem{definition}{Definition}
\newtheorem{proposition}{Proposition}
\newtheorem{lemma}{Lemma}
\newtheorem{theorem}{Theorem}
\newtheorem{corollary}{Corollary}

\newcommand{\be}{\begin{compactenum}}
\newcommand{\ee}{\end{compactenum}}

\newcommand{\store}{\textit{store}\xspace}
\newcommand{\computewitness}{\textnormal{compute-witness}}

\title{Complexity and Unwinding for Intransitive Noninterference%
\footnote{This paper extends and significantly revises the paper
\cite{EMSW11}.  The main differences are 
that full proofs of results from \cite{EMSW11} are provided, 
new results concerning the notion of ITO-security are added, 
new results on unwindings for IP-security and TA-security are added, 
and these new unwindings are used as the basis for new algorithms that
yield better complexity bounds than presented in \cite{EMSW11}.} }

\newif\ifallproofs
\allproofstrue

\newcommand{\authorrefmark}[1]{$^#1$}
\author{
    Sebastian Eggert\authorrefmark{1} ~~~ 
    Ron van der Meyden\authorrefmark{2} \\
    Henning Schnoor\authorrefmark{1} ~~~~~~~~~~~~
    Thomas Wilke\authorrefmark{1} \\[15pt]
1: Institut f\"ur Informatik, \\
    Kiel University \\ 
 2: School of Computer Science and
    Engineering, \\University of New South Wales\\[15pt]
    Email: sebastian.eggert@email.uni-kiel.de, \\
    meyden@cse.unsw.edu.au, \\
    henning.schnoor@email.uni-kiel.de, \\
    thomas.wilke@email.uni-kiel.de
    }

\begin{document} 

\maketitle 

\begin{abstract} 
The paper considers several definitions of information flow security
for intransitive policies from the point of view of the complexity of
verifying whether a finite-state system is secure. 
The results are as follows.
Checking    
\begin{inparaenum}[(i)]
\item P-security (Goguen and Meseguer),
\item IP-security (Haigh and Young), and
\item TA-security (van der Meyden)
\end{inparaenum}
are all in \textbf{PTIME}, while checking TO-security 
(van der Meyden) is undecidable, as is checking ITO-security (van der Meyden).
 The most important ingredients in the proofs of the \textbf{PTIME} upper bounds are new characterizations of the respective security 
notions, which also
lead to new unwinding proof techniques that are shown to be sound and complete for these notions of security, and 
enable the algorithms to 
return simple counter-examples demonstrating insecurity. Our results for IP-security improve
a previous 
doubly  exponential 
bound of Hadj-Alouane et al.  

~
\end{abstract}

\section{Introduction}

One of the fundamental methods in the construction of secure systems to high levels of assurance is to decompose the system
into trusted and  untrusted components,  arranged in an architecture that constrains the possible 
causal effects and 
flows of information between these components.
On the other hand,  resource limitations  and cost constraints  may make it desirable  for trusted and untrusted components   to share  resources.
For example,  it is cheaper for an intelligence  analyst  to handle   high security and low security information on a single desktop machine  than to use
two physically separated  machines.  This leads to complex systems designs and implementations,  in which the desired constraints on flows of information between trusted and untrusted components  need to be enforced in spite of the fact that these components share resources.  In order to provide high levels of assurance of implementations of this kind,   it is desirable to have a formal  theory of systems architecture and information flow, so that a design or implementation  may be formally verified  to conform to an information flow policy. 
Moreover, one would like, whenever 
possible, to automate the verification that a system satisfies such a formally defined policy. 
This motivates the problems we consider in this paper. We study the complexity of verification of a range of formally defined security policies that 
specify how a system is architecturally structured in terms of how information may flow between its components. 

{\bf Attack model:} The problems we consider in this paper address systems implementation attacks.
We work in the paradigm of information flow security, where it is assumed that a (passive) adversary may attack the system by
attempting to make subtle deductions from her possible observations of the system, exploiting covert channels that may exist in the system,
in order to learn secrets that she is not authorized to possess. The automated analyses we consider aim to provide assurance that
the system has been designed  in such a way that such attacks are not possible, or to discover such attacks when they exist. The analysis can be applied both in circumstances where it is feared that a rogue  systems developer may have deliberately constructed the system so as to contain
such prohibited flows of information, as well  as to ensure that such flows of information have not been inadvertently allowed to exist.

{\bf Policy model:} 
Notions of {\em noninterference}---a first definition was given by Goguen and Meseguer \cite{GogMes}---are one approach to the formalisation of information flow and causal relationships.   
Noninterference was first proposed in the context of transitive information flow  policies (with transitivity following from the
partial order on security domains) but it was subsequently noted
\cite{HY87}  that  systems architectures often require  intransitive  policies.  For example, a common architectural  pattern   is to restrict information flow  from a high-level domain to a low-level domain  so as to be possible only via a trusted downgrader (e.\,g., a declassification guard or encryption device).  This pattern motivates  an intransitive information flow policy, stating  that information flow is permitted from the high-level domain to the downgrader  and from the downgrader to the low-level domain,  but not directly from the high-level domain to the low-level domain.

Goguen and Meseguer's definition of noninterference, based on a ``purge'' function, does  not  yield the desired conclusions for intransitive policies.  Haigh and  Young proposed a variant 
(that we refer to as \ipsecty)
for intransitive  policies based on an ``intransitive purge'' function.
Rushby \cite{rushby92} later refined their theory and developed connections to access control systems.
Van der Meyden \cite{meyden2007} has argued that the definitions of security for intransitive policies in these works suffer from  some subtle flaws, and
proposed some improved definitions, TA-security,  
TO-security and ITO-security, 
that first build an operational (full information protocol) model of the maximal permitted information flow in the system, and then compares  the actual information flow to this maximal permitted information flow. The revised
definitions can be shown to avoid the subtle flaws in the intransitive purge-based definition, and  lead to a more satisfactory proof theory and connection
to access control systems than in Rushby's work (e.g., yielding both soundness and completeness results, whereas Rushby proved only
soundness.)

{\bf Verification:} The goal of high assurance systems development by
formal verification motivates the investigation of techniques whereby a systems
design or implementation can be formally shown to satisfy a formal definition of security.
The technique of unwinding relations \cite{goguen84,rushby92} provides a proof method that has been applied to
establish that a system satisfies noninterference properties, but it requires significant human ingenuity to define an
unwinding relation that forms the basis for the proof, and typically also has involved manual driving (proof rule selection)
of the theorem proving tool  within which the proof is conducted.

A better alternative, more acceptable to engineers when it can be applied, is for the property
to be verified by fully automatic techniques. There is a substantial body of work on automated verification techniques
for transitive noninterference properties (which we discuss in Section~\ref{sec:related}), but there has been significantly less work on automated verification
techniques for intransitive  noninterference properties.

{\bf Contributions:} Our contribution in this paper is to provide a basis for automated verification of definitions of intransitive noninterference,
by developing a characterization of the computational complexity of deciding whether a given finite-state system
is secure with respect to an intransitive information flow policy according to this definition.
In particular, we consider Goguen and Meseguer's purge-based definition, 
\ipsecty, 
and van der Meyden's definitions of TA-security,  
TO-security and ITO-security.
We show that the last 
two of these definitions are undecidable,
but the others are decidable in polynomial time and even in nondeterministic logarithmic space. 
We give algorithms for the decidable cases and analyse their complexity. 
Our results are based on new characterizations of IP-security and TA-security. 
Using these new characterizations, we develop new notions of unwinding for IP-security and TA-security, 
that give sound and complete proof techniques, and yield polynomial time decision 
procedures for these two notions. 
Our PTIME decision procedures exploit the new notions of unwinding. 
These new notions of unwinding are 
also 
of independent interest, in that they apply not just to the finite 
state case, where we give the complexity bounds, but also to 
infinite state systems, where they can form the basis for proof theoretic  verification 
methods. 

The structure of the paper is as follows. 
In Section~\ref{sec:defs} we define the formal systems model that we work with,  
and recall the formal definitions of security for intransitive information flow policies that we study. 
New characterizations of 
\ipsecty\ 
and van der Meyden's notion of TA-security 
are presented in Section~\ref{sec:charn}. 
The new unwinding relations for IP-security and TA-security are developed in Section~\ref{sec:unwinding}.
Section~\ref{sec:complexity} gives the complexity results for 
all the security notions that we consider.
Our results are positioned within the literature in Section~\ref{sec:related}, where we discuss related work. 
Section~\ref{sec:concl} concludes with a discussion of open problems and future research directions. 
A reduction that deduces the undecidability of ITO-security from the undecidability of 
TO-security is presented in an appendix.  


\section{Basic Definitions and Notation} 
\label{sec:defs}

In this section, we introduce intransitive information flow 
policies and describe their motivation. We present a
deterministic asynchronous systems model in which such policies
may be interpreted, and then recall a number of different 
semantic interpretations of such policies in this system 
model that have been proposed in the literature. 

\subsection{Noninterference Policies}

Noninterference policies 
are reflexive relations ${\nintrel} \subseteq \Dom\times \Dom$, where 
$\Dom$ is a set of ``domains". The intuitive reading of $u\nintrel v$ 
is that 
``actions of domain $u$ are permitted to interfere with domain $v$'',
or ``information is permitted to flow from domain $u$ to domain $v$''.
For any set $U \subseteq \Dom$ the image of $U$, denoted $\nintrelimage{U}$,  is defined by $\nintrelimage{U} =
\{ v \in \Dom \mid \exists u \in U : u \nintrel v \}$. For a singleton
set $\{ u\}$ we also write $\nintrelimage{u}$ instead of
$\nintrelimage{ \{u\} }$.

The reason for the assumption of reflexivity is that, intuitively, 
a domain should be allowed to interfere with or have information about itself, 
since this cannot usually be prevented.  
 In early work on noninterference \cite{GogMes}, 
 the relation $\nintrel$ is also assumed to be transitive. This follows 
 from the interpretation of domains as corresponding to 
 security levels associated to classes of information and access rights, 
 which have generally been taken to be partially ordered \cite{Denning}.  
(In  the classical multi-level security models, this partial order is derived 
from a linear order on security levels and the set containment order 
on sets of labels.) 

One of the motivations for the consideration of 
policies $\nintrel$ that are not transitive is that
classical multilevel security policies are too restrictive 
for practical purposes, allowing flow of information 
from lower security levels to higher security levels,
but prohibiting flow in the opposite direction. 
Such flows may be less frequent but are nevertheless required, 
e.g., for distribution of battle plans,  in response to 
freedom of information requests, 
or for transmission of encrypted content across an insecure network.
One of the ways this has been handled is to allow the 
general policy to be violated by a special downgrader component. 
A typical downgrader policy is  depicted in Figure~\ref{fig:downgrader}. 
Here the usual (transitive) multi-level policy for domains Public, Secret and 
Top-Secret is extended by the addition of two domains DownS and DownP, 
that are responsible for downgrading of information from Top-Secret to 
Secret, and from Secret to Public, respectively. 
These domains are {\em trusted} to enforce whatever policy constraints apply to the 
downgrading of information. Note that it would not be appropriate to 
apply an assumption of transitivity on this setting, since then, e.g., the edges
involving 
DownS
would imply that Top-Secret $\nintrel$ Secret, i.e., 
a direct flow of information from Top-Secret to Secret is permitted. 

\begin{figure} 
\centerline{\includegraphics[width=6.5cm]{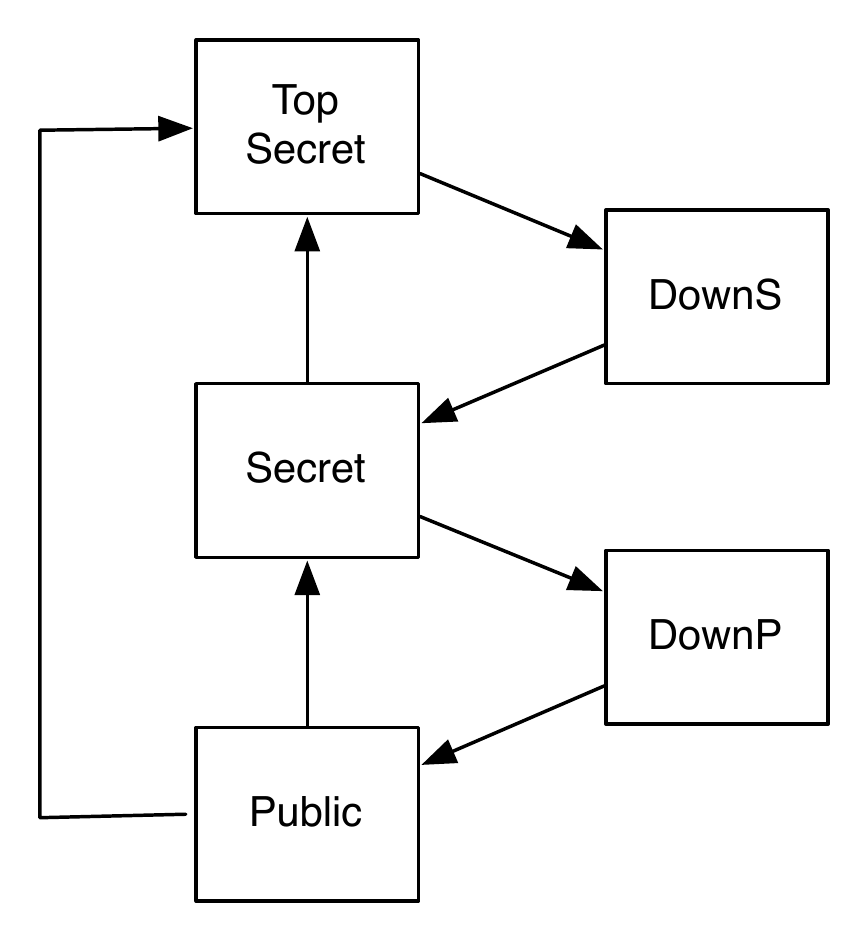}}
\caption{Downgrader Policy\label{fig:downgrader}}
\end{figure} 

Subsequent work on 
intransitive
 noninterference
has taken a somewhat extended interpretation of  the term 
``domain,'' treating this more as akin to ``component" in a
systems  architecture. Figure~\ref{fig:crypto} shows a 
systems architecture, discussed in \cite{rushby92} and \cite{BDRF08},
for a system in which messages are sent from a high security 
(Red) domain through a low security (Black) domain, with the 
global security policy stating that  all content, except the message header, 
must be encrypted, and uncontrolled flow of information from 
Red to Black is prohibited. The architecture proposes to achieve 
this goal by having the Bypass component check a (more detailed) policy 
on the allowed header structure and content, and by having 
the Crypto component enforce a local policy stating that all output must be 
encrypted.  These flows are recomposed into the encrypted message 
(with header)
at the Black component.
Crypto and Bypass are assumed to be trusted components of low enough 
complexity that they can be verified to enforce their local policies. 
Red (which may contain Trojans) and Black (which is at a low security level) 
are not assumed to be trusted.
The argument for security of the system is intended to follow from the structure of the
information flows in the architecture, plus the  assumption that 
the trusted components correctly implement their local policies.

\begin{figure} 
\centerline{\includegraphics[width=6.5cm]{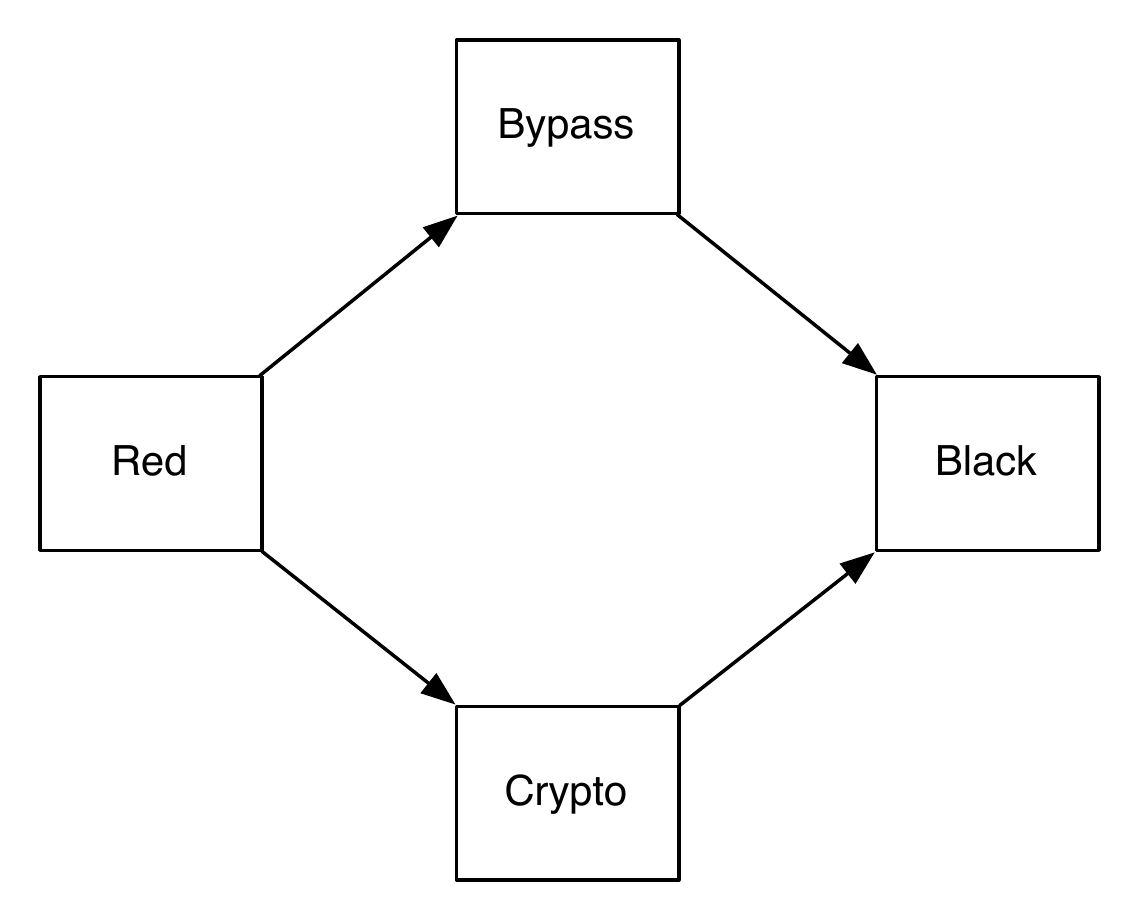}}
\caption{Policy for Encrypted Message Transmission\label{fig:crypto}}
\end{figure} 

MILS security, as expounded in  \cite{BDRF08}, proposes to 
base development of certifiably secure systems on 
design level arguments of this type, together with implementations
in which mechanisms such as separation kernels or periods processing
are used to enforce the systems architecture. 
We refer to \cite{BDRF08} for a more detailed discussion of MILS
security and the proposed structure of the argument for security 
of the system in Figure~\ref{fig:crypto}.

\begin{figure} 
\centerline{\includegraphics[width=5cm]{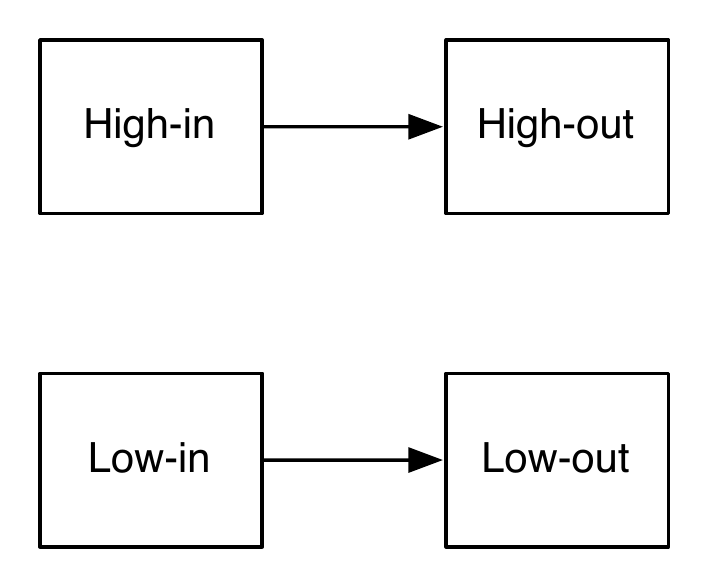}}
\caption{A MILS Design Level Policy\label{fig:multiplex}}
\end{figure}

\begin{figure}[t]
\centerline{\includegraphics[width=6cm]{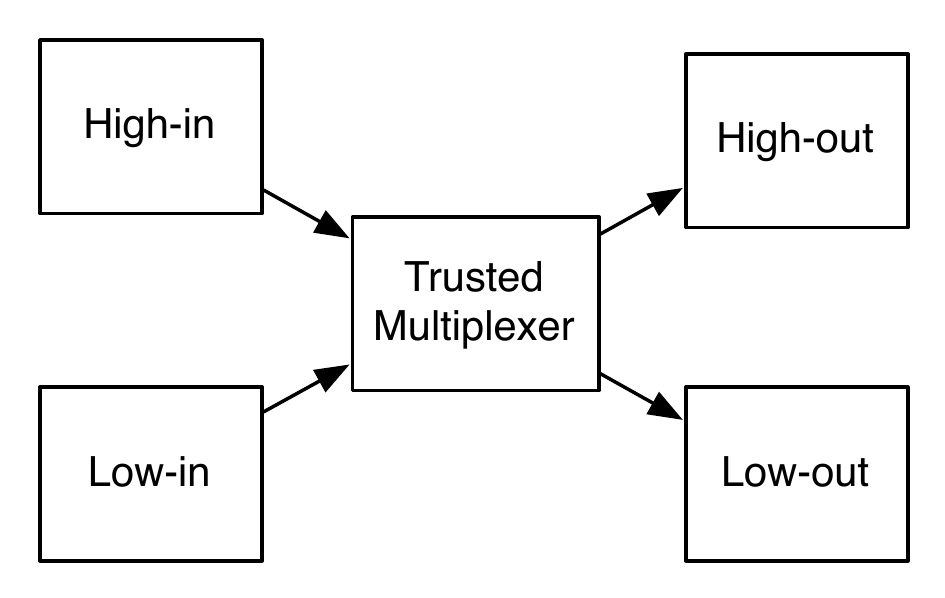}}
\caption{A MILS Resource Sharing Implementation\label{fig:multiplex2}}
\end{figure}

We note that intransitive information flow policies are intended to
express just the architectural structure of information flow, 
rather than encompass all the details of security policy. 
One 
key point is that implementations may involve resource sharing, 
which may mean that it is not immediately apparent that 
the design level architecture is enforced in the implementation. 

For example, Figure~\ref{fig:multiplex}  illustrates a design level policy for a system 
with multiple independent security levels that 
could be implemented, as shown in Figure~\ref{fig:multiplex2}, 
by a trusted multiplexer component that handles information  from multiple security levels. 
One of the issues in the verification of such systems is to determine whether such a resource sharing implementation 
correctly enforces the design level architecture. The definitions in the following sections
provide a number of distinct semantic interpretations of information flow policies 
that have been proposed to formalize what it means to implement 
the notion of correct enforcement.

\subsection{State-Observed Machine Model} 

Several different types of semantic models have been used in the literature on noninterference. 
(See \cite{MZ10} for a comparison and a discussion of their relationships.)
We work here with the state-observed machine model 
used by Rushby~\cite{rushby92}, but similar results would be obtained 
for other models.

This model consists of deterministic machines of the form
$\langle \States ,s_0, \Actions, \step, \obs,\linebreak\dom \rangle$, where $\States $ is a set
of states, $s_0\in \States $ is the \emph{initial state}, $\Actions$ is a set
of actions, $\dom\colon  \Actions \rightarrow \Dom$ associates each action
with an element of the set $D$ of security domains, $\step\colon \States \times
\Actions \rightarrow \States $ is a deterministic transition function, and
$\obs\colon \States \times D \rightarrow O$ maps states to an observation in some set
$O$, for each security domain. We may also refer to security domains 
more succinctly as ``agents''. We write $s\cdot
\alpha$ for the state reached by performing the sequence of actions
$\alpha\in \Actions^*$ from state $s$, defined inductively by $s\cdot
\epsilon = s$, and $s\cdot \alpha a
= \step(s\cdot \alpha, a) $ for $\alpha\in \Actions^*$ and $a\in
\Actions$. Here, $\epsilon$ denotes the empty sequence.
For any string $\alpha$, we say a symbol $a$ occurs in $\alpha$ if $\alpha = \beta a
\beta'$ for some strings $\beta, \beta'$. We
define $\Alphabet{\alpha}$ as the set of all symbols occurring in
$\alpha$.

\subsection{The Purge Function}

Noninterference is given a formal semantics in the transitive case
\cite{GogMes}
using a definition based on a ``purge'' function. Given a set $E\subseteq D$ of 
domains and a sequence $\alpha \in \Actions^*$, we write 
$\spurge{\alpha}{E}$ for the subsequence of all actions $a$ in $\alpha $
with $\dom(a) \in E$. 
Given a policy $\nintrel$, 
we define the function $\purge\colon \Actions^* \times \Dom \rightarrow \Actions^*$ 
by \[\purge(\alpha,u) = \spurge{\alpha}{\{v\in D\mid v\nintrel u\}}.\]  
(For clarity, we may use subscripting of
agent arguments of functions, writing, e.\,g.,  $\purge(\alpha,u)$ as
$\purge_u(\alpha)$.) The system $M$ is said to be 
\emph{secure with respect to the transitive policy $\nintrel$},
when, for all $\alpha\in \Actions^*$ and domains $u\in
\Dom$, we have $\obs_u(s_0\cdot \alpha) = \obs_u(s_0\cdot
\purge_u(\alpha))$. That is, each agent's observations are as if only
interfering actions had been performed.  An equivalent formulation 
(which we state more generally for policies that are not necessarily 
transitive, in anticipation of later discussion) is
the following:

\medskip

\begin{definition}[\psecty]
  A system $M$ is \emph{\psec} with respect to a policy $\nintrel$
  if for all sequences $\alpha, \alpha'\in \Actions^*$ such that
  $\purge_u(\alpha) = \purge_u(\alpha')$, we have $\obs_u(s_0\cdot
  \alpha) = \obs_u(s_0\cdot \alpha')$. 
\end{definition} 

\medskip

This can be understood as saying
that agent $u$'s observation depends only on the sequence of
interfering actions that have been performed. 

\subsection{The Intransitive Purge Function}

\begin{figure} 
\centerline{\includegraphics[width=8.5cm] {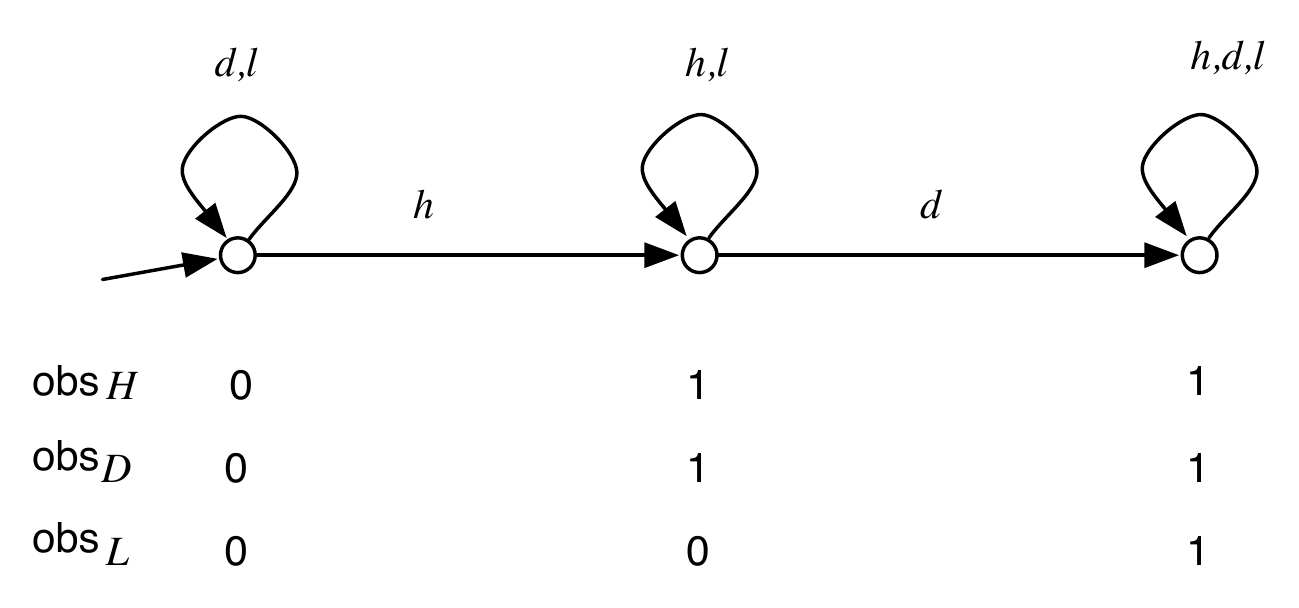}} 
\caption{A system that is \tosec\ but not \psec\label{fig:notpsec}}
\end{figure} 

While \psecty\ is a reasonable definition of security for transitive information flow policies, 
it works less well for intransitive policies. Figure~\ref{fig:notpsec} illustrates a 
system that is, intuitively, secure for the downgrader policy $H\nintrel D\nintrel L$, but which does not satisfy 
\psecty. Here $h,d,l$ are actions of domains $H,D,L$, respectively, and the observations in each domain 
are depicted below the states. 
 Intuitively, the observations convey a single bit of information: ``has $H$ ever performed the action $h$?''.   
Domains $H$ and $D$ learn that $H$ has performed $h$ as soon as this action is performed
(by their observations turning to value $1$), but $L$ does not learn this until after $D$ subsequently 
performs the downgrading action $d$. Since the policy permits $D$ to transmit 
information about $H$, the system is secure. 
However, this system does not satisfy \psecty, since we have $\purge_L(hdl) = dl = \purge_L(dl)$ but 
$\obs_L(s_0\cdot hdl) = 1 \neq 0 = \obs_L(s_0\cdot \purge_L(dl))$.
 Intuitively, \psecty\ says that $L$ observations
depend only on what $D$ and  $L$  actions have been performed, so cannot
contain information about $H$, even though the policy, intuitively, permits $D$ to 
transmit such information. 

To address this deficiency,
Haigh and Young~\cite{HY87} generalized the definition of the purge function to 
intransitive policies.
Intuitively, 
 the intransitive purge of a 
sequence of actions with respect to a domain $u$ is the largest subsequence
of actions that could form part of a causal chain of effects
(permitted by the policy) ending with an effect on domain $u$. 
More formally
(we follow the presentation from \cite{rushby92}),
the
definition makes use of a function $\sources\colon  \Actions^*\times \Dom
\rightarrow \pow{\Dom}$ defined inductively by $\sources(\epsilon,u) =
\{u\}$ and, for $a\in \Actions$ and $\alpha\in \Actions^*$, if there exists $v \in \sources(\alpha,u)$ with $\dom(a) \nintrel v$, then 
\begin{align*}
  \sources(a \alpha,u) 
  = \sources(\alpha,u) \cup \{\dom(a)\} \enspace,
\end{align*}
and else
\begin{align*}
  \sources(a \alpha,u) = \sources(\alpha,u) \enspace.
\end{align*}
Intuitively, $\sources(\alpha,u)$ is the set of domains $v$
such that there exists a sequence of permitted interferences 
from $v$ to $u$ within $\alpha$. The {\em intransitive purge} 
function $\ipurge\colon  \Actions^* \times \Dom \rightarrow \Actions^*$ 
is then defined inductively by 
$\ipurge(\epsilon,u) = \epsilon$ and, for $a\in \Actions$ and $\alpha\in \Actions^*$, if $\dom(a) \in \sources(a\alpha,u)$, then
\begin{align*}
\ipurge(a \alpha,u) = a \, \ipurge(\alpha,u) \enspace,
\end{align*}
and else
\begin{align*}
 \ipurge(a \alpha,u) = \ipurge(\alpha,u) \enspace.
\end{align*}

The intransitive purge function is then used in place of the purge function in 
Haigh and Young's definition: 

\medskip

\begin{definition}[\ipsecty] 
A system $M$ is \ipsec\ with respect to a (possibly intransitive) policy $\nintrel$
if for all sequences $\alpha\in \Actions^*$, and $u\in D$, we have 
$\obs_u(s_0\cdot \alpha) = \obs_u(s_0\cdot \ipurge_u(\alpha))$.
\end{definition} 

\medskip

Since the function $\ipurge_u$ on $\Actions^*$ is idempotent, 
this definition, like the definition for the transitive case, 
can be formulated as: 
$M$ is \ipsec\ with respect to a policy $\nintrel$ if
for all $u\in D$ and all sequences $\alpha, \alpha'\in \Actions^*$
with $\ipurge_u(\alpha) = \ipurge_u(\alpha')$, we have 
$\obs_u(s_0\cdot \alpha) = \obs_u(s_0\cdot \alpha')$.
It can be seen that $\ipurge_u(\alpha) = \purge_u(\alpha)$ when 
$\nintrel$ is transitive, so \ipsecty\ is in fact a generalisation of the 
definition of security for transitive policies. 

\subsection{The \ta\ Function}

\begin{figure} 
\centerline{\includegraphics[width=8.5cm]{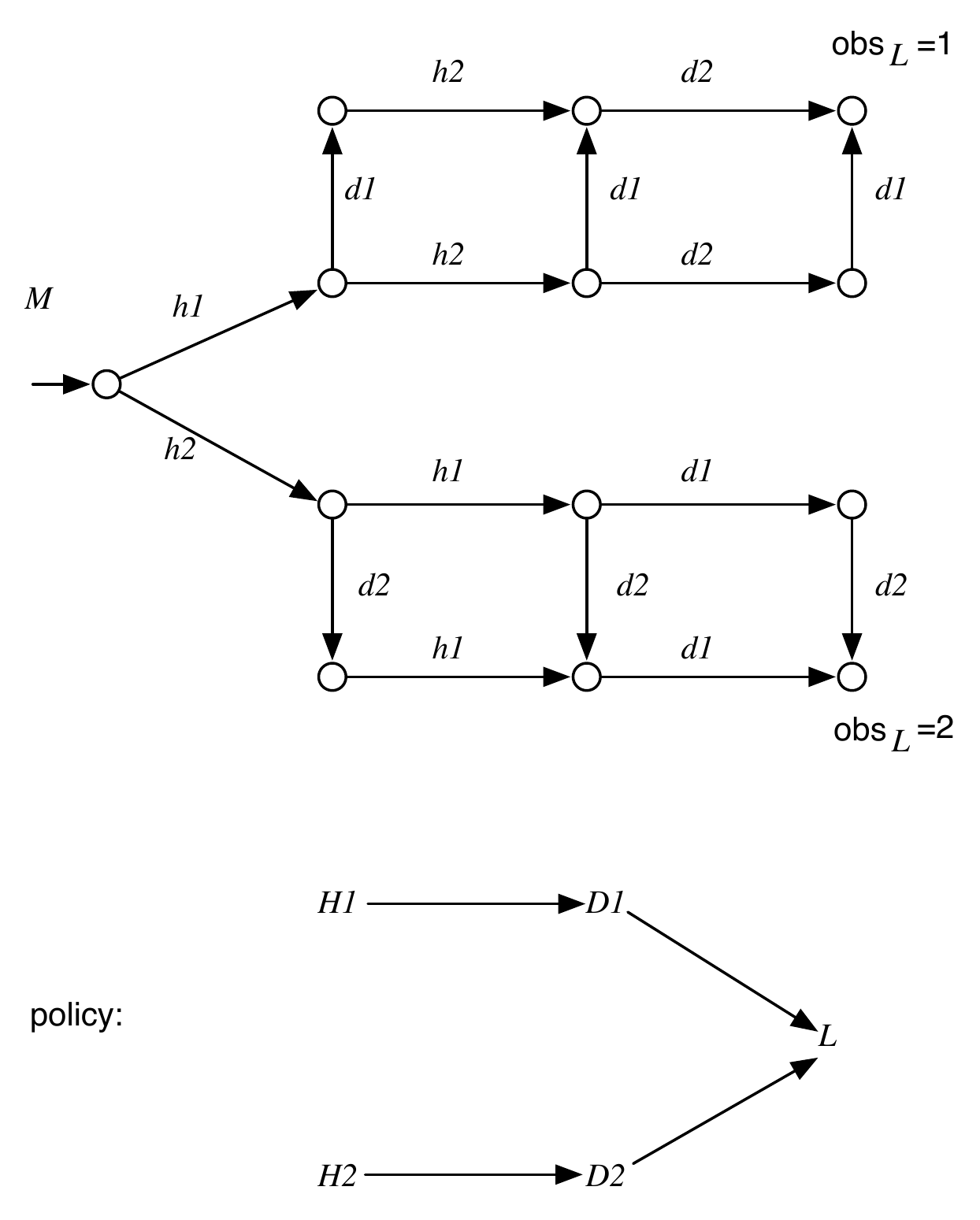}}
\caption{A system that is IP-secure but not TA-secure\label{fig:not-tasec}} 
\end{figure} 

It has been noted by van der Meyden \cite{meyden2007} that
\ipsecty\ classifies some systems as secure where there is, intuitively, 
an insecure flow of information that relates to a domain learning ordering
information about the actions of other domains that it should not have.  

Figure~\ref{fig:not-tasec} depicts 
part of 
a system $M$ and a policy $\nintrel$ 
such that $M$ is \ipsec, but for which the conclusion that the system is secure is questionable. 
We sketch the argument for this here, and refer the reader to \cite{meyden2007} for a more
rigorous presentation.  Intuitively, the system is comprised of two High security level domains $H1,H2$, each 
with a downgrader 
($D1,D2$, respectively)
 to the Low security domains $L$. 
The actions $h1,h2,d1,d2$ are associated to the domains $H1,H2,D1,D2$, respectively, and
state transitions are depicted only when there is a change of state. The observations of 
$L$ are depicted at two  of the states; at all other states we assume that $L$ makes
observation $0$. All other agents may be assumed to make observation $0$ at all states. 
Intutively, at 
the state where  $L$ observes  $1$, it is possible for $L$ to deduce that 
there has been an occurrence of $h1$ followed by an occurrence of $h2$; 
the state where  $L$ observes  $2$, it is possible for $L$ to deduce
that these actions have occurred in the opposite order. 

We show that this system is \ipsec:
Suppose we have $\ipurge_L(\alpha) = \ipurge_L(\beta)$, 
and one of $\obs_L(s_0\cdot \alpha)$ or   $\obs_L(s_0\cdot \beta)$ is 
$1$ or $2$, say the former is equal to 1.  Then this sequence must contain
an occurrence of $h1$ before an occurrence of $h2$, and each is followed
by $d_1$ and $d_2$, respectively. This observation shows, in fact, that 
$L$ {\em knows} the order of  the first $h_1$ and $h_2$ actions in the sequence $\alpha$.  
Because $\ipurge_L$ preserves $h1$ when it is followed by $d1$, 
and similarly for $h2$ and $d2$, and also preserves the order of 
actions that it retains,  the same statement must hold for $\beta$, 
and it then follows that also $\obs_L(s_0\cdot \beta) =1= \obs_L(s_0\cdot \alpha)$. If neither observation is in $1,2$, then both are
equal to $0$, and again we have the required equality of observations. 

On the other hand, the conclusion that the system is secure is somewhat peculiar. 
Each of the downgraders is individually permitted by the policy to know only about 
activity in its associated High level domain, and its own activity. Thus, individually, 
neither $D1$ nor $D2$ can know the order of the first  two $H1$ and $H2$ actions. 
Moreover, since the system is asynchronous, even if we were to combine all 
the information that the downgraders are permitted to know, 
we would still not be able to deduce the order on the $H1,H2$ actions.  
 We therefore have the peculiar conclusion that the system is classified
 by \ipsecty\ to be secure, but it allows $L$ to learn information 
 that would not be permitted to be known to the two domains $D1,D2$, 
 which are 
 supposed to filter all flow of information from $H1,H2$,  even if these 
 domains were to combine their information. 

To address this
peculiarity, 
van der Meyden has proposed some other interpretations
of intransitive policies.   Both proceed by first defining a concrete operational model of the 
maximal amount of information that an agent is permitted to 
have after some sequence of actions has been performed. 
Security of the system is then defined by  requiring 
that an agent's observation may not contain more than this maximal amount of information. 

In the first operational model, 
when an agent performs an action, it transmits what it is permitted to 
know to other agents, subject to constraints in the policy. The following
definition expresses
this 
 in a weaker way than the ipurge function.

Given sets $X$ and $\Actions$, let the set  
$\hset(X,\Actions)$ 
be the smallest set containing $X$  
and such that if $x,y\in \hset$ and $z\in \Actions$  then $(x,y,z) \in \hset$. 
Intuitively, the elements of $\hset(X,\Actions)$ are binary trees with 
leaves labelled from $X$ and interior nodes labelled from $\Actions$. 

Given a policy $\nintrel$, define, for each agent $u\in \Dom$,  the function $\ta_u\colon  \Actions^* \rightarrow 
\hset(\{\epsilon\},\Actions)$ inductively by $\ta_u(\epsilon) = \epsilon$, 
and, for $\alpha\in \Actions^*$ and $a\in \Actions$, 
\begin{equation*}
\ta_u(\alpha a)  =  
\begin{cases}
  \ta_u(\alpha) & \text{if $\dom(a) \not\nintrel u$,} \\
  (\ta_u(\alpha), \ta_{\dom(a)}(\alpha), a) & \text{otherwise.}
\end{cases}
\end{equation*}
Intuitively, $\ta_u(\alpha)$ captures the maximal information that agent $u$ may, 
consistently with the policy $\nintrel$, have 
about the past actions of other agents. 
Initially, an agent has 
no
information about what actions have been performed. 
The recursive clause describes how the maximal information $\ta_u(\alpha)$ 
permitted to flow to $u$ after the performance of $\alpha$ changes when the 
next action $a$ is performed. If $a$ may not interfere with $u$, then there 
is no change, otherwise, $u$'s maximal permitted information is 
increased by adding the maximal information permitted to $\dom(a)$
at the time $a$ is performed 
(represented by $\ta_{\dom(a)}(\alpha)$), 
as well the fact that $a$ has been performed. Thus, this 
definition captures the intuition that an agent may only transmit information that it 
is permitted to have, and then only to agents with which it is permitted to 
interfere.

\medskip

\begin{definition}[\tasecty]
A system $M$ is  \tasec\ with respect to a policy $\nintrel$ if
for all agents $u$ and all $\alpha,\alpha'\in \Actions^*$ such that $\ta_u(\alpha) = \ta_u(\alpha')$, 
we have $\obs_u(s_0\cdot \alpha) = \obs_u(s_0\cdot \alpha')$.
\end{definition} 

\medskip

Intuitively, this says that each agent's observations provide the 
agent with no more than 
the maximal amount of information that may have been transmitted to it, 
as expressed by the functions $\ta$.

\subsection{The \pt\ Function}

In the definition of TA-security, the operational model of information flow 
given by the function $\ta$ permits a domain to transmit  information  
that it {\em may} have, even if it has never observed anything 
from which it could deduce that information. Arguably, 
this is too liberal. 

\begin{figure} 
\centerline{\includegraphics[width=8.5cm] {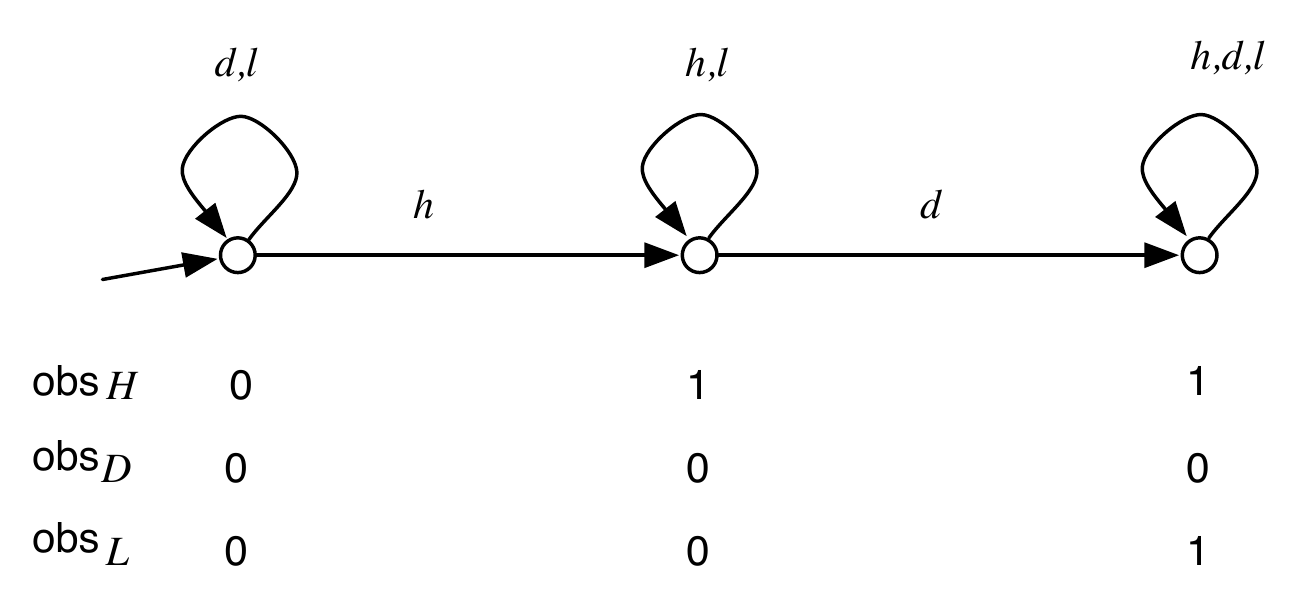}} 
\caption{A system that is \tasec\ but 
neither \tosec\ not \itosec\label{fig:nottosec}}
\end{figure} 

Figure~\ref{fig:nottosec} shows a system
for the downgrader policy $H \nintrel D\nintrel L$, 
similar to that in  Figure~\ref{fig:notpsec}. 
It can be argued that the system is \tasec; we leave the 
details to the reader. 
Again, when $L$ observes  $1$, it can deduce that 
$H$ has performed the action $h$, and indeed, this  observation can 
only occur after $D$ has performed the action $d$, thereby 
downgrading the information about $H$. 
On the other hand, note that in this system, $D$'s 
observation is always $0$, so $D$ cannot know, on the
basis of its observations, whether $H$ has performed $h$. 
$D$ is therefore transmitting to $L$ information that it 
does not itself have. 

Van der Meyden \cite{meyden2007} therefore also considers
a variant operational model in which a domain transmits
only what it has actually observed. This yields the alternate notion of 
\tosecty. 

The sequence of all observations and actions of a domain is denoted as
its view. Formally, the notion of \emph{view} is defined as follows. 
The definition uses an absorptive
concatenation function $\acat$, defined over a set $X$ by
$s\acat x = s$ if $x$ is equal to the final
element of $s$ (if any), and $s\acat x = s\cdot x$ (ordinary
concatenation) otherwise, for every $s\in
X^*$ and $x\in X$.  Define the view of domain $u$ with respect
to a sequence $\alpha\in \Actions^*$ using the function $\view_u\colon
\Actions^*
\rightarrow (\Actions \cup \Observations)^*$
(where $\Observations$ is the set of observations in the system) defined by 
\begin{align*}
\view_u(\epsilon) & = \obs_u(s_0)\text{, and}\\
 \view_u(\alpha a) &= (\view_u(\alpha) \cdot b
)\acat \obs_u(s_0\cdot \alpha) \enspace,
\end{align*}
 where $b =a$ if $\dom(a) =u$ and $b =
\epsilon$ otherwise.  That is, $\view_u(\alpha)$ is the sequence of
all observations and actions of domain $u$ in the run generated by
$\alpha$, compressed by the elimination of stuttering observations.
Intuitively, $\view_u(\alpha)$ is the complete record of information
available to agent $u$ in the run generated by the sequence of actions
$\alpha$. The reason we apply the 
absorptive 
concatenation is to capture that the
system is asynchronous, with agents not having access to a global
clock. 
The effect of this operation is to reduce any stuttering 
of an observation in the run to a single copy. 
Thus, two sequences that only differ from each other in repetitions of
a single observation are not distinguishable by the agent.

Given a policy $\nintrel$, for each domain $u\in \Dom$, define the function 
$\pt_u\colon \Actions^*\rightarrow \hset((\Actions\cup \Observations)^*,\Actions)$
by $\pt_u(\epsilon) =\obs_u(s_0)$ and
\begin{equation*}
\pt_u(\alpha a) \\
= 
\begin{cases}
\pt_u(\alpha) & \text{if $\dom(a) \not \nintrel u$,}\\ 
  (\pt_u(\alpha), \view_{\dom(a)}(\alpha), a) &  \text{otherwise.}
\end{cases}
\end{equation*}
Intuitively, this definition takes the model of the maximal information 
that an action $a$ may transmit after the sequence $\alpha$ to be the fact that $a$ has occurred, 
together with the information that $\dom(a)$  {\em actually} has, as represented by 
its view $\view_{\dom(a)}(\alpha)$.  By  contrast, \tasecty\ uses in place of this the maximal information 
that $\dom(a)$ {\em may} have. 
We may now base the definition of security on the 
function $\pt$ rather than $\ta$.

\medskip

\begin{definition}[\tosecty]
The system $M$ is \tosec\ with respect to $\nintrel$ if
for all domains $u\in D$ and all  $\alpha, \alpha'\in \Actions^*$ with $\pt_u(\alpha) = \pt_u(\alpha')$, 
we have $\obs_u(s_0\cdot\alpha) = \obs_u(s_0\cdot \alpha')$.
\end{definition} 

\medskip

It is possible to give a flatter representation of the information in
$\pt_u(\alpha)$ 
that clarifies the relationship of this definition to \psecty. 
Define the \emph{possibly transmitted view} of domain $u$ for a sequence of
actions $\alpha$ to be the largest prefix $\tview_u(\alpha)$ of
$\view_u(\alpha)$ that ends in an action $a$ with $\dom(a) =u$.
Then we have the following result, 
which intuitively says that $u$'s observations depend only on 
(1) the parts of the views of other agents which are permitted to 
pass information to $u$ that they have actually acted to transmit, and 
(2) $u$'s knowledge of the ordering of its own actions and the actions 
of these other agents. 

\medskip

\begin{proposition} \emph{(Characterization of \tosecty\ \cite{meyden2007})} 
  \label{prop:ptflat}
The system $M$ is \tosec\ with respect to a policy
$\nintrel$ iff for all sequences $\alpha,\alpha' \in \Actions^*$, and
domains $u\in \Dom$, if 
$\purge_u(\alpha) = \purge_u(\alpha')$ 
and $\tview_v(\alpha) = \tview_v(\alpha')$ for all domains $v\neq u$
such that $v\nintrel u$, then $\obs_u(s_0\cdot \alpha) =
\obs_u(s_0\cdot\alpha')$.
\end{proposition}

\subsection{The \ito\ Function}

In order to compare with a definition of Roscoe and Goldsmith \cite{RoscoeG99}, 
van der Meyden has also introduced a variant of \tosecty\ called \itosecty, in which 
information is transmitted slightly faster. We also consider this 
notion here since our complexity results bear on algorithmic claims of Roscoe and 
Goldsmith. 

Given a policy $\nintrel$, for each domain $u\in D$, define the function 
$\ito_u:\Actions^*\rightarrow \hset(O(A\cup O)^*,A)$
by $\ito_u(\epsilon) =\obs_u(s_0)$ and
\[
\ito_u(\alpha a) =\left\{
\begin{array}{ll}
\ito_u(\alpha) & \mbox{if $\dom(a) \not \nintrel u$,}\\ 
  (\ito_u(\alpha), \view_{\dom(a)}(\alpha), a) &  \mbox{if $\dom(a) = u$,}\\
  (\ito_u(\alpha), \view_{\dom(a)}(\alpha a), a) &  \mbox{otherwise.}
  \end{array}\right.\] 
This definition is just like that of $\pt$, with the difference that
the information that may be transmitted to $u$ by an action $a$ 
such that $\dom(a) \nintrel u$ but $\dom(a) \neq u$, includes the
observation $\obs_{\dom(a)}(s_0\cdot \alpha a)$ obtained in domain 
$\dom(a)$ immediately after the occurrence of action $a$. 
Intuitively, the definition of security based on this notion will
allow that the action $a$ transmits not just the information
observable to $\dom(a)$ at the time that it is invoked, but also the
new information that it computes and makes observable in $\dom(a)$.
This information is not included in the value
$\ito_{\dom(a)}(\alpha a)$ itself, since the definition of security
will state that the  new observation may depend only on this value.
The nomenclature in this case is intended to be suggestive of {\em
immediate} transmission of information about observations.

The following definition follows the pattern of the
others, but based is on the functions $\ito$. 
\begin{definition} 
The system $M$ is \itosec\ with respect to $\nintrel$ if
for all domains $u\in D$ and all  $\alpha, \alpha'\in \Actions^*$ with 
$\ito_u(\alpha) = \ito_u(\alpha')$, 
we have $\obs_u(s_0\cdot\alpha) = \obs_u(s_0\cdot \alpha')$. 
\end{definition} 

\medskip

Figure~\ref{fig:itonottosec} gives an example of a system, for the downgrader 
policy $H \nintrel D\nintrel L$, that is ITO-secure, but not TO-secure. 
Intuitively, in the action sequence $hd$, the downgrader
learns that $h$ has been performed (by making observation $1$)
from the observation that it makes after performing the action $d$. 
ITO-security permits that the information in this observation is transmitted
to $L$ by the action $d$, whereas TO-security does not. 

\begin{figure} 
\centerline{\includegraphics[width=8.5cm] {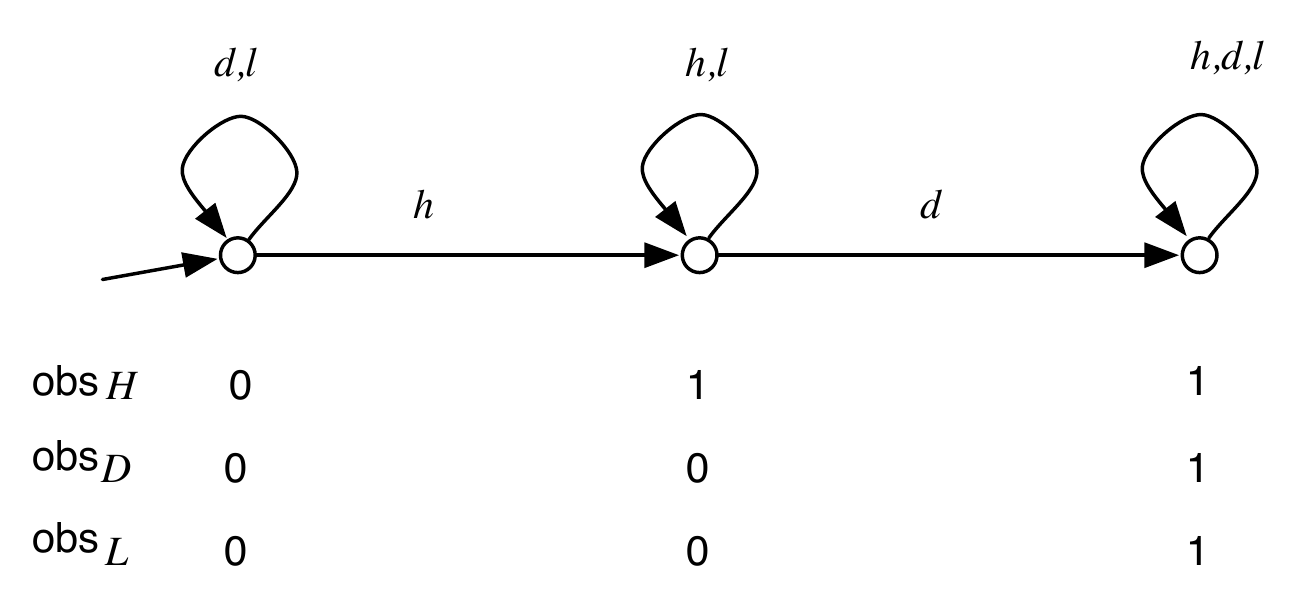}} 
\caption{A system that is \itosec\ but not \tosec\label{fig:itonottosec}}
\end{figure}

The definitions introduced above are
shown in  \cite{meyden2007} to be related as
follows: \psecty\ implies \tosecty\ implies 
\itosecty\ implies 
\tasecty\ implies \ipsecty.
The converse of each of these implications does not hold: 
Figures~\ref{fig:notpsec}-\ref{fig:itonottosec} 
provide counter-examples. 
In the special case of transitive policies  $\nintrel$, all these 
notions are equivalent.

\section{Characterization of \ipsecty\ and \tasecty}
\label{sec:charn} 

In this section, we develop new characterizations of \ipsecty\ and \tasecty, 
that enable the new unwindings and decision procedures for these notions 
of security.

\subsection{Characterization of \ipsecty}

We present a new characterization of \ipsecty. This characterization
is the main tool for our later algorithm that verifies \ipsecty\ in polynomial time. 

Intuitively, the ipurge function that defines \ipsecty\ removes actions that should be irrelevant for
the domain $u$ from its ``visible trace.''  This leads us to the definition of the relation $\RedIrr_u$: 
for $u \in \Dom$ and $\alpha,\alpha'\in \Actions^*$, we define $\alpha \RedIrr_u \alpha'$ if 
$\ipurge_u(\alpha ) = \ipurge_u(\alpha')$ and 
there exist $\beta, \beta' \in \Actions^*$, and $a \in \Actions$ such that 
$\alpha = \beta a \beta'$ and $\alpha' = \beta \beta'$. 
That is, $\alpha \RedIrr_u \alpha'$ if $\alpha'$ is
obtained from $\alpha$ by removing a single action that is ``irrelevant'' in
the sense that (according to the information flow allowed by the policy) $u$
should not be able to observe whether the removed action has occurred at all.
The symmetric closure of $\RedIrr$ is denoted with $\SymRedIrr$.

Note that if $\ipurge_u(\alpha) = \ipurge_u(\alpha')$, 
then there exists a sequence $\alpha = \alpha_0 \SymRedIrr \alpha_1\SymRedIrr \ldots   
\SymRedIrr \alpha_n = \alpha'$. If $\obs_u(s_0 \cdot  \alpha) \neq \obs_u(s_0 \cdot
\alpha')$, then we must have   $\obs_u(s_0 \cdot  \alpha_k) \neq \obs_u(s_0 \cdot
\alpha_{k+1})$ for some $k$. Thus, 
directly from the definition of \ipsecty, 
we obtain
that a system is IP-insecure iff
there exists a domain $u \in \Dom$, a reachable state $q \in \States$, $a \in \Actions$ and
$\alpha \in \Actions^*$ such that $\ipurge_u(a \alpha) =
\ipurge_u(\alpha)$ and $\obs_u(q \cdot a \alpha) \neq \obs_u(q \cdot
\alpha)$. 
We now state a lemma that shows that we can put some restrictions on
$\alpha$. 
The lemma shows that if a system is not \ipsec, then an $\alpha$ and $a$ as above
exist such that additionally, $\alpha$ does not contain any action $c$ whose domain $\dom(c)$ 
can be influenced by $\dom(a)$. This allows us to reduce the search space for a witness of insecurity significantly when designing our algorithms.

 In the following, we will always assume that every state $s\in S$ is reachable,
i.e., there is a sequence $\alpha\in\Actions^*$ such that $s_0\cdot\alpha=s$.

\medskip

\begin{lemma}
\label{lemma:char-ipsecty}
  A system $M$ is IP-insecure iff there exist $u \in \Dom$, $q \in \States$, $a \in \Actions$ and
$\alpha \in \Actions^*$ such that 
\begin{asparaenum}[(i)]
\item $\ipurge_u(a \alpha) = \ipurge_u(\alpha)$,
\item $\obs_u(q \cdot a \alpha) \neq \obs_u(q \cdot \alpha)$, and
\item $\nintrelimage{\dom(a)} \cap \{ \dom(c) | c \in \Alphabet{\alpha} \} =
  \emptyset$.
\end{asparaenum}
\end{lemma}

\ifallproofs

\begin{proof}
  Let  a system $M$ be IP-insecure. 
  Then there exist
  $u \in \Dom$, $q \in \States$, $a \in \Actions$ and
$\alpha \in \Actions$ such that $\ipurge_u(a \alpha) =
\ipurge_u(\alpha)$ and $\obs_u(q \cdot a \alpha) \neq \obs_u(q \cdot
\alpha)$. 
We fix this domain $u$ and choose $q, a, \alpha$ such that $| \alpha
|$ is minimal for all choices of $q, a, \alpha$ satisfying 
$\ipurge_u(a \alpha) =
\ipurge_u(\alpha)$ and $\obs_u(q \cdot a \alpha) \neq \obs_u(q \cdot
\alpha)$. 

Assume that there exists $b \in \Actions$ such that $\alpha = \beta b
\beta'$ for some $\beta, \beta' \in \Actions^*$ with $\dom(a) \nintrel
\dom(b)$. 
We show that such an action $b$ can't exist by considering the
following three cases: 
\begin{enumerate}
\item  $\obs_u(q \cdot \beta b \beta') \neq \obs_u(q
  \cdot \beta \beta')$:
  We set $q' = q\cdot \beta$, $a' = b$ and $\alpha' = \beta'$. 
  Then the condition $\ipurge_u(b \beta') = \ipurge_u(\beta')$ is
  satisfied and we could choose $q', a', \alpha'$ instead of $q, a,
  \alpha$. This contradicts the minimal length of $\alpha$, since $|
  \alpha' | < | \alpha |$.
\item $\obs_u(q \cdot a \beta b \beta') \neq \obs_u(q
  \cdot a \beta \beta')$:
  We set $q' = q \cdot a \beta$, $a'= b$ and $\alpha' = \beta'$. 
  With the same argument as in the previous case, we get a
  contradiction.
\item $\obs_u(q \cdot \beta b \beta') = \obs_u(q
  \cdot \beta \beta')$ and $\obs_u(q \cdot a \beta b \beta') = \obs_u(q
  \cdot a \beta \beta')$: 
  This gives
  \begin{align*}
\obs_u(q
  \cdot \beta \beta')  =
\obs_u(q \cdot \beta b \beta') 
& \neq \obs_u(q \cdot a \beta b \beta') 
 & = \obs_u(q
  \cdot a \beta \beta') \enspace. 
  \end{align*}
 We set $q' = q$, $a'= b$ and $\alpha' = \beta\beta'$. 
 Again, the condition $\ipurge_u(\beta\beta') = \ipurge_u(a \beta
 \beta')$ is satisfied.
 We get a contradiction since $| \alpha'| < | \alpha |$.  
\end{enumerate}
\end{proof}\ 

In fact, with nearly the same proof, one can show that, if the length
of $\alpha$ is minimal for all choices of $q, a, \alpha$, then $\alpha
= \ipurge_u(\alpha)$. 
\fi

\subsection{Characterization of \tasecty}\label{sect:ta characterization}
We present a new characterization of \tasecty\ that also makes precise its relationship
to \ipsecty. As seen earlier, the latter is concerned with the question which actions an agent
$u$ may observe at all, hence $\ipurge_u(\alpha)$ is obtained from $\alpha$ by removing
from $\alpha$ actions that should be ``unobservable'' for $u$, provided that information
only flows as specified by the security policy. 
The definition of \tasecty\ in~\cite{meyden2007} was motivated by the observation that the 
security-relevant information that 
should be unobservable to some agents  is not just  which actions appear at all,
but also information about the \emph{order} in which certain actions are performed. 
This type of information-flow is not prohibited by the definition of \ipsecty.

In this section we 
show 
that what separates the definition of \ipsecty\ from that of \tasecty\ 
is  the question how much information is known about execution orders of
actions.  
\tasecty\ can essentially be seen as \ipsecty\ 
plus the requirement that an agent
should only have access to ``timing information'' (i.e., information about the order
of the occurrence of actions) insofar as permitted by the security policy.

To formalize this, we require a few technical definitions.
The following definition captures the situation in which an agent $u$ should not have information
about the order in which certain actions are performed, although 
it 
may know 
whether these actions have been performed, and how often.

\medskip

\begin{definition}[swappable]
\label{def:swappable}
Let $\alpha, \alpha' \in \Actions^*$ and $a, b \in \Actions$ and $u \in
  \Dom$. We write $\alpha a b \alpha' \Swap_u \alpha b a \alpha'$ iff
  $\nintrelimage{\dom(a)} \cap \nintrelimage{\dom(b)} \cap \{ u,
  \dom(c) | c \in \Alphabet{a b \alpha'} \}  = \emptyset$.  
In this case, we call the actions $a$ and $b$ \emph{swappable} in
$\alpha ab\alpha'$. 

For any relation $\rightarrow$, we define $\stackrel = \rightarrow$ as
the reflexive closure of $\rightarrow$ and $\stackrel * \rightarrow$
as the reflexive, transitive closure of $\rightarrow$.
\end{definition}

We will see later that Definition~\ref{def:swappable} captures exactly the issue mentioned above: 
If $\alpha ab\alpha'\Swap_u\alpha ba\alpha'$, then the action sequences $\alpha ab\alpha'$ and
$\alpha ba\alpha'$ should be indistinguishable for agent $u$, even though 
it 
is allowed to know
whether actions $a$ and $b$ have been performed. The reason why, intuitively, $u$ should not have access
to this ``timing information'' is that only agents $w\in\nintrelimage{\dom(a)}\cap\nintrelimage{\dom(b)}$ can directly observe whether $a$ or $b$ is performed first. If no agent that can observe this information directly performs any action in $\alpha$, then, after performing $\alpha$, the agent $u$ should not have this information either (unless of course, $u$ is in the intersection.)

We call strings $\alpha, \alpha' \in \Actions^*$ \emph{order
  indistinguishable for $u$}, 
and write $\alpha \Symrefswap_u \alpha'$,  if $\alpha \stackrel * \leftrightarrow^\textsl{swap}_u \alpha'$.  

\ifallproofs
We note that the definition of swappable could be relaxed to also allow swaps of actions $a$ and $b$
when the agents that observe both $a$ and $b$ only perform actions in $\alpha$ that cannot be observed
by $u$ via the policy---however we will later only apply these definitions to action sequences to
which $\ipurge$ has already been applied. We therefore use the above definition to simplify notation.
\fi
The following lemma shows that our definition correctly captures the above intuition. It states that
information about the order of ``swappable'' operations are indeed hidden from an agent by the
definition of \tasecty.

\medskip

\begin{lemma}
\label{lemma:swap swappable elements}
 Let $u \in \Dom$,  $\alpha, \alpha' \in \Actions^*$ with $\alpha
 \Swap_u \alpha'$, then $\ta_u(\alpha) = \ta_u(\alpha')$.
\end{lemma}

\medskip

\ifallproofs

\begin{proof}
Let be $\beta, \beta' \in \Actions^*$ and $a, b \in \Actions$ such
that
$\alpha = \beta a b \beta'$ and $\alpha' = \beta b a \beta'$. 
We proceed with an induction on the length of $\beta'$. 
First we assume that
$\beta a b \Swap_u \beta b a$ and  show for all 
$u' \not\in \nintrelimage{\dom(a)} \cap
\nintrelimage{\dom(b)}$ that $\ta_{u'}(\beta a b) = \ta_{u'}(\beta b a)$
holds. 
Note that by definition of $\Swap_u$ we have $u, \dom(a), \dom(b)
\not\in \nintrelimage{\dom(a)} \cap \nintrelimage{\dom(b)}$. 
Without loss of generality, assume that $u'\in \nintrelimage{\dom(a)}
\cup \nintrelimage{\dom(b)}$, since otherwise 
$\ta_{u'}(\beta a b) = \ta_{u'}(\beta) = \ta_{u'}(\beta b a)$. 
Also without loss of generality assume that $u' \in
\nintrelimage{\dom(a)}$. 
Therefore $\dom(b) \not\nintrel \dom(a)$ and $\dom(b) \not\nintrel
u'$.
This gives
\begin{multline*}
  \ta_{u'}(\beta a b) = \ta_{u'}(\beta a)  = ( \ta_{u'}(\beta),
  \ta_{\dom(a)}(\beta), a) \\
  = (\ta_{u'}(\beta b), \ta_{\dom(a)}(\beta b), a) = \ta_{u'}(\beta b
  a) \enspace. 
\end{multline*}
Assume that $\alpha c \Swap_u \alpha' c$ for some $c \in \Actions$. 
And we assume that inductively, for any agent $u'\in \Dom$:
\begin{equation*}
  \text{If } \alpha \Swap_{u'} \alpha', \text{ then } \ta_{u'}(\alpha)
  = \ta_{u'}(\alpha') \enspace. 
\end{equation*}
The definition of $\Swap$ gives
\begin{align*}
  \emptyset  = & \nintrelimage{\dom(a)} \cap \nintrelimage{\dom(b)}
  \cap
  \{ u, \dom(d) | d \in \Alphabet{a b \beta' c} \} \\
   = & \nintrelimage{\dom(a)} \cap \nintrelimage{\dom(b)} 
  \cap ( \{ u, \dom(d) | d \in \Alphabet{a b \beta'} \} \\
  & \phantom{\cap (\ } \cup \{ \dom(c), \dom(d)
  | d \in \Alphabet{a b \beta'} \} )\\
  = & (\nintrelimage{\dom(a)} \cap \nintrelimage{\dom(b)} 
  \ \ \ \ \cap \{ u, \dom(d) | d \in \Alphabet{a b \beta'} \} )  \\
  & \phantom{\cap (\ } \cup  
  (\nintrelimage{\dom(a)} \cap \nintrelimage{\dom(b)} 
  \cap \{ \dom(c), \dom(d)
  | d \in \Alphabet{a b \beta'} \} ) \enspace. 
  \end{align*}
Therefore $\beta a b \beta' \Swap_{\dom(c)} \beta b a \beta'$, and from the
prerequisites we also know that 
$\beta a b \beta' \Swap_u \beta b a \beta '$.
Applying the induction hypothesis gives $\ta_u(\alpha) =
\ta_u(\alpha')$ and $\ta_{\dom(c)}(\alpha) = \ta_{\dom(c)}(\alpha')$. 
If $\dom(c) \not\nintrel u$ the we get directly,
\begin{equation*}
  \ta_u(\alpha c) = \ta_u(\alpha) = \ta_u(\alpha') = \ta_u(\alpha' c)
  \enspace. 
\end{equation*}
In the case of $\dom(c) \nintrel u$ we get: 
\begin{align*}
  \ta_u(\alpha c)  = (\ta_u(\alpha), \ta_{\dom(c)}(\alpha), c) 
  & = (\ta_u(\alpha'), \ta_{\dom(c)}(\alpha'), c) \enspace. 
\end{align*}
\end{proof}
\fi

The following corollary combines the above result and the fact that \tasecty\ implies \ipsecty: 

\medskip

\begin{corollary}
  \label{corollary:symrefswap}
  Let be $u \in \Dom$ and $\alpha, \alpha' \in \Actions^*$ with $\ipurge_u(\alpha)
  \Symrefswap_u \ipurge_u(\alpha')$, then $\ta_u(\alpha) = \ta_u(\alpha')$.
\end{corollary} 

\medskip

\ifallproofs
\begin{proof}
 From the definition of \tasecty\  it follows (see also the proof of Theorem~1 in 
 the full version of~\cite{meyden2007}) that $\ipurge_u(\alpha)=\ipurge_u(\alpha')$ implies 
 $\ta_u(\alpha)=\ta_u(\alpha')$. The corollary now follows from this
 result and 
repeated  
 application of Lemma~\ref{lemma:swap swappable elements}.
\end{proof}
\fi

We now 
state 
the result mentioned earlier: \tasecty\ is, in a very precise sense, \ipsecty\
plus the requirement that agents should not be able to distinguish between action sequences
that are order indistinguishable. The following theorem shows that the information that an agent
is not permitted to have in the definition of \ipsecty, in addition to information already
forbidden to 
it
by \ipsecty, is exactly the information about the orders of actions
that are ``swappable.''

\medskip

\begin{theorem}
  \label{thm:ta-swap}
  Let be $u \in \Dom$ and $\alpha, \alpha' \in \Actions^*$. 
  Then $\ta_u(\alpha) = \ta_u(\alpha')$ if and only if
  $\ipurge_u(\alpha) \Symrefswap_u \ipurge_u(\alpha')$.
\end{theorem}

\medskip

\ifallproofs
\begin{proof}
  If $\ipurge_u(\alpha) \Symrefswap_u \ipurge_u(\alpha')$, then $\ta_u(\alpha)=\ta_u(\alpha')$ follows from Corollary~\ref{corollary:symrefswap}. 

  For the other direction, assume that $\ta_u(\alpha) = \ta_u(\alpha')$. 
  Since $\ipurge_u(\alpha)=\ipurge_u(\alpha')$ implies $\ta_u(\alpha)=\ta_u(\alpha')$
  (Corollary~\ref{corollary:symrefswap}), we can without loss of generality assume that $\alpha =
  \ipurge_u(\alpha)$ and $\alpha' = \ipurge_u(\alpha')$. 
  We also assume $\alpha\neq\alpha'$. Note that the number of occurrences of an action $a$ in $\alpha$ is
  the same as in $\alpha'$. 
  Let be $\alpha'' \in \Actions^*$ such that $\alpha''\Swap_u \alpha'$
  and $\alpha''$ has a common prefix with $\alpha$ of maximal length among all $\alpha''$ with this property. 
  We can write $\alpha = \beta a \beta'$ and $\alpha'' = \beta
  \beta''$ with $\beta, \beta', \beta'' \in \Actions^*$, $a \in
  \Actions$ and $a \beta' \neq \beta''$. 
  We also assume that the position of $a$ in $\beta''$ is the
  left-most position among all possible choices of $\alpha''$. 
  The $\beta''$ is of the form $\gamma  b a \gamma'$ with $\gamma b a
  \gamma' \not\Swap_u \gamma a b \gamma'$ and $\gamma, \gamma' \in
  \Actions^*$, $b \in \Actions$. 
  Therefore there is some agent $u' \in \sources_u(\gamma')$ with
  $\dom(a) \nintrel u'$ and $\dom(b) \nintrel u'$. 
  Therefore the tree $\ta_{u'}(\beta\gamma b a)$ is a subtree of $\ta_u(\alpha')$. 
  Because of the corresponding number of occurrences of $a$ in $\alpha$, the corresponding subtree would be $\ta_{u'}(\beta a)$. 
  But the number of occurrences of $b$ in this two trees does not match. 
  This contradicts the assumption that $\alpha' \neq \alpha''$. 
\end{proof}
\fi

The characterization obtained by the above theorem is now stated in the following corollary:

\medskip

\begin{corollary}\label{corollary:ta sec characterization 2}
A system $M$ is \tasec\ if and only if it is \ipsec\ and for every state $q$, every agent $u$,
and every $a,b\in\Actions$, $\alpha\in\Actions^*$, if $a$ and $b$ are swappable in $ab\alpha$,
then $\obs_u(q\cdot ab\alpha)=\obs_u(q\cdot ba\alpha)$.
\end{corollary}

\medskip

\ifallproofs
\begin{proof}
We first show that if a system is not \tasec, then it violates the condition.
 Theorem~\ref{thm:ta-swap} says that for all action sequences $\alpha$ and
$\alpha'$ we have $\ta_u(\alpha) = \ta_u(\alpha')$ if and only if there exist actions $\alpha_0, \dots, \alpha_n$ such that 

\begin{align*}
\alpha & \SymrefRedIrr \alpha_0 \SymrefRedIrr \dots \SymrefRedIrr
\alpha_k \\
& \RefSwap \alpha_{k+1} \RefSwap \dots \RefSwap \alpha_l \\
& \SymrefRedIrr \alpha_{l+1} \SymrefRedIrr \dots \SymrefRedIrr \alpha_n \SymrefRedIrr \alpha' \enspace. 
\end{align*}

For the converse, we assume that $M$ is \ipsec\ and not \tasec. Then there exist traces $\alpha$ and $\alpha'$
and a state $q$ such that $\obs_u(q\cdot\alpha)\neq\obs_u(q\cdot\alpha')$, and $\ta_u(\alpha)=\ta_u(\alpha')$.
Let $\alpha_0$, $\alpha_1$, \dots\ be the sequence of action sequences as above. Since $M$ is \ipsec, it follows that
if $\alpha_i\SymRedIrr\alpha_{i+1}$, then $\obs_u(q\cdot\alpha_i)=\obs_u(q\cdot\alpha_{i+1})$. Therefore, the above
implies that there are traces $\alpha_i$ and $\alpha_{i+1}$ such that $\alpha_i\RefSwap\alpha_{i+1}$ and
$\obs_u(q\cdot\alpha_i)\neq\obs_u(q\cdot\alpha_{i+1})$ as claimed.

Now assume that the system $M$ is \tasec. In~\cite{meyden2007}, it was shown that \tasecty\ implies \ipsecty.
It remains to prove that $M$ satisfies the condition of the corollary.
If the condition is not satisfied, then there
s a state $q$, agent $u$, and actions $a,b\in\Actions$ and $\alpha\in\Actions^*$ such that $a$ and $b$
are swappable in $ab\alpha$, and $\obs_u(q\cdot ab\alpha)\neq\obs_u(q\cdot ba\alpha)$.
Since Theorem~\ref{thm:ta-swap} implies that $\ta_u(ab\alpha)=\ta_u(ba\alpha)$, it follows
that the system is not \tasec, a contradiction.
\end{proof}
\fi

\section{Unwindings} \label{sec:unwinding} 

The characterizations discussed above provide the 
basis for complexity results for each of the notions of security. In the case
of P-security, IP-security and TA-security, our complexity results are obtained by 
means of an appropriate notion of unwinding. For P-security, this is the classical 
notion of unwinding for transitive policies, but for IP-security and TA-security, 
we develop, in this section, novel notions of unwinding that we show to be both sound and complete. 
These notions are of independent interest in that they provide proof methods that can be 
applied to establish security (e.g., by use of theorem provers) even when the decision procedures
are inapplicable (e.g., because the system is infinite-state, or beyond the practical scope of 
our decision procedures.) The complexity results are given in the following section. 

For the remainder of this section, we assume that all states of a system $M$ are reachable, noting 
that all our notions of security hold in $M$ iff they hold in the reachable part of $M$. 

We first recall the notions of unwinding defined in  \cite{rushby92}. 
Given a system $M$, an {\em unwinding} with repect to a policy $\nintrel$ is a collection of equivalence  
relations $\unwind_u$ on the set of states of $M$, for $u\in \Dom$, satisfying the following conditions, 
for all $u\in \Dom$, $a\in \Actions$ and $s,t\in S$:  
\begin{itemize}
\item[\ocp:] If $s \unwind_u t$ then $\obs_u(s) = \obs_u(t)$.
\item[\scp:] If $s \unwind_u t$  then $s \cdot a \unwind_u  t \cdot
  a$.
\item[\lrp:] If 
  $\dom(a) \not \nintrel u$ then $s \unwind_u s\cdot
  a$.
\end{itemize}
Rushby shows that, for {\em transitive} policies $\nintrel$ (where  all our notions of security coincide \cite{rushby92,meyden2007}), a system $M$ is P-secure with respect to $\nintrel$  iff there exists an unwinding on $M$.   
In fact, this result applies even for intransitive policies, with exactly the same proof.  

Rushby also defines a {\em weak unwinding} with respect to a policy $\nintrel$ to be a similar collection 
of relations that satisfy \ocp, \lrp and the following weaker version of \scp: 
\begin{itemize}
\item[WSC:] If $s \unwind_u t$ and $a \in \Actions$ and $s\unwind_{\dom(a)} t$ 
then $s \cdot a \unwind_u t \cdot  a$. 
\end{itemize} 
It is shown in \cite{rushby92} that if there exists a weak unwinding on a system $M$ with respect to $\nintrel$, then 
$M$ is IP-secure with respect to $\nintrel$. Subsequently, it was shown in \cite{meyden2007} that the 
conclusion can be strengthened to TA-security of $M$, and also that if $M$ is TA-secure then there
exists a weak unwinding on the unfolding of $M$, a system that behaves exactly like $M$ except that it 
records  in the system state the sequence of actions performed. However, it 
is {\em not} the case that there is always a weak unwinding on the original state space of a TA-secure system.

\subsection{An Unwinding for IP-security}

We now define a new notion of unwinding that we show to be sound and complete for 
IP-security. 

We define an unwinding relation $\unwind_u^v$ 
on the states of the system $M$
for every 
pair of domains  
$u, v \in
\Dom$.
The 
domain 
 $u$ takes the part of 
 the observer. 
The 
domain 
$v$ 
corresponds to a domain whose actions are not supposed to directly affect the 
observations of $u$. 
We say that the collection of equivalence relations 
$\unwind^v_u$ for $u,v\in \Dom$ is an {\em IP-unwinding on $M$ with respect to $\nintrel$}, 
if for all domains $u,v$, states $s,t$ of $M$, and actions $a\in A$, the following hold: 
\begin{itemize}
\item[\ocip:] If $s \unwind_u^v t$ then $\obs_u(s) = \obs_u(t)$.
\item[\scip:] If $s \unwind_u^v t$ and
  $v  \not\nintrel  \dom(a)$ then $s \cdot a \unwind_u^v t \cdot
  a$.
\item[\lrip:] If $v \not \nintrel u$ and 
  $\dom(a) = v$ then $s \unwind_u^v s\cdot
  a$.
\end{itemize}
With 
this definition it is possible to give a full
characterisation of \ipsecty. 

\begin{theorem} \label{thm:ip-unwind}
  A system $M$ is \ipsec\ iff 
there exists an IP-unwinding on $M$ with respect to $\nintrel$.
\end{theorem}

\begin{proof}
We first prove the implication from left to right. 
Let $M$ be an \ipsec \ system. For $u, v \in \Dom$, 
  we define the following relation:
  \begin{equation*}
    s \unwind_u^v t \text{ iff } \forall \alpha \in 
    \{ a \in \Actions | v \not \nintrel \dom(a) \}^*: 
    \obs_u( s \cdot \alpha) = \obs_u(t \cdot \alpha) \enspace.  
  \end{equation*}
  We show that this collection of relations is an IP-unwinding on $M$ with respect to $\nintrel$.  
  
  Obviously, $\unwind_u^v$ is an equivalence relation. 
   For any $s, t \in \States$, \ocip\ is trivially satisfied by taking the empty
  string for $\alpha$. 
  It is left to show that $\unwind_u^v$ satisfies the conditions 
  \scip\ and \lrip. 

  For \scip, let  $a \in \Actions$ with $v \not \nintrel
  \dom(a)$.
  By definition of $s \unwind_u^v t$ we have, for all $\alpha \in \{ c
  \in \Actions | v \not \nintrel \dom(c) \}^*$, that $ \obs_u(s \cdot
  \alpha) = \obs_u(t \cdot \alpha)$. 
  This implies that, for all  $\alpha \in \{ c
  \in \Actions | v \not \nintrel \dom(c) \}^*$, we have $ \obs_u(s \cdot a
  \alpha) = \obs_u(t \cdot a \alpha)$. 
  Again by definition, this gives  $s\cdot a  \unwind_u^v t\cdot a$. 
  
  For \lrip\ we assume that $v \not \nintrel u$.
  Let $s\in S$ and $a \in \Actions$ with $\dom(a) = v$. 
  Since we assume that all states are reachable, there exists a sequence of actions $\beta$ such that 
  $s = s_0\cdot \beta$. 
  Then by applying \ipsecty\ of $M$: 
  \begin{align*}
    \forall \alpha \in 
    \{ c \in \Actions | v \not \nintrel \dom(c) \}^*: 
    \obs_u( s \cdot \alpha) & = 
    \obs_u( s_0 \cdot \beta \alpha) \\ 
  &  =  \obs_u(s_0 \cdot \ipurge_u(\beta \alpha)) \\
    & = \obs_u(s_0 \cdot \ipurge_u(\beta a \alpha)) \\
        & = \obs_u(s_0 \cdot \beta a \alpha) \\
    & = \obs_u(s \cdot a \alpha) \enspace. 
  \end{align*}
  This shows that $s \unwind_u^v s \cdot a$.

  For the other direction of this proof we assume that $M$ is
  IP-insecure, and show that there does not exist an IP-unwinding. 
  By Lemma~\ref{lemma:char-ipsecty}, 
   there exists a reachable state $q \in \States$,  $u \in \Dom$, $a \in
  \Actions$ and $\alpha \in \Actions^*$ such that
  $\ipurge_u(a \alpha) = \ipurge_u(\alpha)$ and $\obs_u(q \cdot a
  \alpha) \neq \obs_u(q \cdot \alpha)$ and $\nintrelimage{\dom(a)}
  \cap \{ \dom(c) | c \in \Alphabet{\alpha} \} = \emptyset$. 
  Set $v = \dom(a)$ and let  $\unwind_u^v$ an equivalence relation on $\States$ that
  satisfies \scip\ and \lrip. 
  We derive a contradiction to \ocip, showing that no IP-unwinding can exist. 
  By \lrip\  we have $q \unwind_u^v q \cdot a$, since $v \not \nintrel
  u$. 
  By applying \scip\  multiple times, we get 
  $q \cdot \alpha \unwind_u^v
  q \cdot a \alpha$. 
  Since $\obs_u(q \cdot \alpha) \neq \obs_u(q \cdot a \alpha)$, \ocip\ 
  is not satisfied. 
\end{proof}

Note that the equivalence relations satisfying \scip\ and \lrip\ form 
a sublattice of ${\cal P}(S^2)$. This gives the following corollary, that 
will form the basis for the decision procedure for IP-security 
described in Section~\ref{sec:alg-ip}. 

\begin{corollary}
  Let $M$ be a finite system and, for every $u, v \in \Dom$,
   let $\unwind_u^v$
  be the smallest
  equivalence relation that satisfies \scip\  and \lrip\ 
  with respect to $\nintrel$. 
  Then $M$ is \ipsec\ with respect $\nintrel$ to iff $\unwind_u^v$ satisfies \ocip for every $u, v
  \in \Dom$.
\end{corollary}

\subsection{An unwinding for TA-security} 

Next, we characterize TA-security by means of a new notion of unwinding. 

Define a TA-unwinding of system $M$ with respect to policy $\nintrel$ to be 
a collection of equivalence relations
$\unwind_u^{v, w}$ on the set of states of $M$, 
for every $u \in \Dom$ and every $v, w \in \Dom$ with $v \neq w$, 
satisfying the  following conditions, for all states $s,t\in S$ and actions $a,b\in \Actions$:
\begin{itemize}
\item[\octa:] If $s \unwind_u^{v,w} t$ then $\obs_u(s) = \obs_u(t)$.
\item[\scta:] If $s \unwind_u^{v,w} t$ and if $a \in \Actions$ with
  $v \not\nintrel  \dom(a)$ or $w \not\nintrel \dom(a)$
  then 
  $s \cdot a \unwind_u^{v, w} t \cdot a$.
\item[\lrta:] If 
$\dom(a) = v$ and $\dom(b) =
  w$ and $v \not\nintrel w$ and  $w\not\nintrel v$, and 
  either 
  $v \not\nintrel u$ or $w \not\nintrel u$, then $s \cdot ab  \unwind_u^{v, w} s\cdot
  ba$.
\end{itemize}
(If there is just one domain, then the empty collection is taken to be a TA-unwinding.)
The following result characterises TA-security using this notion of unwinding. 

\begin{theorem}
  A system $M$ is \tasec\ 
  with respect to $\nintrel$ 
  iff $M$ is \ipsec\ with respect to $\nintrel$  and 
  there exists a 
TA-unwinding of $M$ with respect to $\nintrel$. 
\end{theorem}

\begin{proof}
We first show the direction from left to right. 
    Let system $M$ be  \tasec\ with respect to $\nintrel$. 
    It follows that $M$ is IP-secure with respect to $\nintrel$, so 
    it suffices to show that there exists a TA-unwinding with respect to $\nintrel$. 
    For this, we define the following relation on the states of $M$, for  $u, v, w \in \Dom$ 
    with $v \neq w$: 
  \begin{equation*}
    s \unwind_u^{v, w} t \text{ iff } \forall \alpha \in 
    \{ a \in \Actions | v \not \nintrel \dom(a) \text{ or }
    w \not\nintrel \dom(a) \}^* \colon 
    \obs_u( s \cdot \alpha) = \obs_u(t \cdot \alpha) \enspace.  
  \end{equation*}
  That the property \octa\ is satisfied 
  is immediate from the choice of $\alpha = \epsilon$. 
  
  To show \scta, let 
  $s, t \in \States$ with $s \unwind_u^{v,
    w} t$ and $a \in \Actions$
  with $v \not \nintrel \dom(a)$ or $w \not \nintrel \dom(a)$. 
 By definition of $s \unwind_u^{v, w} t$ we have for all 
 $\alpha \in \{ c \in \Actions | v \not \nintrel \dom(c) 
 \text{ or } w \not\nintrel \dom(c) \}^*$, that $ \obs_u(s \cdot
  \alpha) = \obs_u(t \cdot \alpha)$. 
  This implies that for all  $\alpha \in \{ c
  \in \Actions | v \not \nintrel \dom(c) \text{ or } w \not \nintrel
  \dom(c) \}^*$, we have $\obs_u(s \cdot a
  \alpha) = \obs_u(t \cdot a \alpha)$. 
  Again by definition, this gives  $s\cdot a  \unwind_u^{v, w}  t\cdot a$. 
  
  For \lrta\ 
  let 
  $a, b \in \Actions$ with $\dom(a) = v$ and $\dom(b) = w$ 
  and $v \not\nintrel w$, $w \not\nintrel v$ and $v \not\nintrel u$ or
  $w\not\nintrel u$. 
  Let 
  $s \in \States$. 
  Since we assume that all states are reachable, there
  exists $\beta \in \Actions^*$ with 
  $s = s_0 \cdot \beta$ . 
  For any 
  $\alpha \in 
    \{ c \in \Actions | v \not \nintrel \dom(c) \text{ or } w \not
    \nintrel \dom(c)  \}^*$
    we have
    \begin{equation*}
      \nintrelimage{v} \cap \nintrelimage{w} \cap \{ u, \dom(c) \mid c
      \in \Alphabet{ab\alpha} \} = \emptyset \enspace. 
    \end{equation*}
    This shows that $a$ and $b$ are swappable in $\beta a b \alpha$. 
    By Lemma~\ref{lemma:swap swappable elements}, 
    we have $\ta_u(\beta a b \alpha) = \ta_u(\beta
    b a \alpha)$. 
  Then by applying \tasecty\ of $M$ we get for any such an $\alpha$:
  \begin{align*}
    \obs_u(s \cdot a b \alpha)  = \obs_u(s_0 \cdot \beta a b \alpha)    
     = \obs_u(s_0 \cdot \beta b a \alpha ) 
     = \obs_u(s \cdot b a \alpha)
  \end{align*}
  This shows that $s \cdot a b  \unwind_u^{v,w} s \cdot b a$.

  For the other direction of this proof we assume that $M$ is
  \ta-insecure, but \ipsec, 
  with respect to $\nintrel$, and show that there does not exist a TA-unwinding on $M$ with respect to $\nintrel$.
  From 
  Corollary~\ref{corollary:ta sec characterization 2}, it follows that 
  there exists
$q \in \States$, $u \in \Dom$, $a, b, \in \Actions$, $\alpha \in
\Actions^*$ 
such that  $a$ and $b$ are swappable in $a b \alpha$ and 
$\obs_u(q \cdot a b \alpha) \neq \obs_u(q \cdot b a \alpha)$. 
It follows that $v \neq w$. 
We set $v = \dom(a)$ and $w = \dom(b)$. 
We suppose that $\unwind_u^{v, w}$ is
an equivalence relation on $\States$ that satisfies \scta\  and
\lrta, 
and show that it cannot also satisfy \octa. 

Since $a$ and $b$ are swappable in $a b \alpha$ it follows directly
that $v \not\nintrel w$, $w \not\nintrel v$ and $v \not\nintrel u$ or
$w \not\nintrel u$. 
Therefore by \lrta, we have $q \cdot a b \unwind_u^{v, w} q \cdot b a$. 
Since $v \not\nintrel \dom(c)$ or $w \not\nintrel\dom(c)$ for all $c
\in \Alphabet{\alpha}$, it follows by \scta\  that 
$q \cdot a  b\alpha \unwind_u^{v, w} q \cdot b a \alpha$.
  Since $\obs_u(q \cdot a b \alpha) \neq \obs_u(q \cdot
   b a \alpha)$, we obtain that 
  \octa\ 
  is not satisfied. 
\end{proof}

Note that by Theorem~\ref{thm:ip-unwind}, the reference to IP-security 
may be replaced by the existence of an IP-unwinding, giving a characterization 
of TA-security that is stated entirely in terms of the existence of unwinding relations.

\section{Complexity} 
\label{sec:complexity} 

In this section, we consider the complexity of algorithmic 
verification of the notions we have discussed, 
in the case of finite state systems. 

We show that three of these notions (\psecty, \ipsecty, and \tasecty) are decidable, and in fact can be decided in polynomial time, and we prove that \tosecty\ 
and ITO-security are 
undecidable.

\subsection{Complexity of \psecty}\label{sec:alg-p}

The algorithm presented in Figure~\ref{fig:alg-psecty} checks \psecty. 
The main idea of the algorithm is to use the unwinding characterization of 
\psecty, and compute the minimal equivalence relations satisfying 
conditions \scp and \lrp and check that these satisfy \ocp. 
The equivalence relations are represented as partitions of the set $S$, 
using the disjoint-set data
structure which provides functions \makeset, \union\ and \findset, 
with low amortized cost per operation
\cite{Tarjan75}. 

\begin{figure}
\begin{algorithm}[H]
	\caption{Decide \psecty}
 \ForEach{ $u \in \Dom$} 
  {
    \tcc{ create a new partition } 
    \ForEach{ $ s \in \States$ } 
    { 
      \makeset($s$) \;
    }
    let $P$ be an empty list \;
    let $\store$ be empty \;
    \tcc{ apply LR conditions } 
    \ForEach{ $s \in \States$ } 
    { 
			\nllabel{alg:lrcond}
      \ForEach{ $a \in \Actions$ with $\dom(a) \not \nintrel u$ } 
      { 
        \If{ $\findset(s) \neq \findset( s \cdot a )$ }
        {
					add $((s\cdot a, s), (s, s), (a, \epsilon))$ to \store \;
					insert $(s\cdot a, s)$ into the list $P$ \;
          $\union(s\cdot a, s)$ \; 
          \If{ $\obs_u(s \cdot a) \neq \obs_u(s)$ }
					{
						\Return $\computewitness(s\cdot a, s, \epsilon, \epsilon)$ \;
					}
        }
      }
    }
    \tcc{ apply SC conditions } 
			\While{ $P \neq \emptyset$ } 
      { 
        take a pair $(s, t)$ out of $P$ \;
        \nllabel{alg:remove}
        \ForEach{ $a \in \Actions$ }
        { 
          \nllabel{alg:foreachloop}
          \If{ $\findset(s \cdot a) \neq \findset(t \cdot a)$ }
          {
						add $((s\cdot a, t\cdot a), (s, t), (a, a))$ to \store \;
						insert $(s \cdot a, t\cdot a)$ into the list $P$ \;
						$\union(s\cdot a, t\cdot a)$ \;
						\If{ $\obs_u(s \cdot a) \neq \obs_u(t \cdot a)$ }
						{
							\Return $\computewitness(s\cdot a, t\cdot a, \epsilon, \epsilon )$ \;
						}
					}
				}
			}
		}
\Return{``secure''}  
\end{algorithm}
\caption{An algorithm for \psecty.\label{fig:alg-psecty}}
\end{figure} 

\begin{figure}
\begin{procedure}[H]
\caption{compute-witness($s$, $t$, $\alpha$, $\beta$)}
 \If{ $s = t$ }
 {	
	find a shortest path $\gamma$ from $s_0$ to $s$ \;
	\Return $(\gamma \alpha, \gamma \beta)$ \;
	}
	\Else 
	{
		choose stored entry $((s, t), (s', t'), (a, b))$ \;
		$\computewitness(s', t', a \alpha, b \beta)$ \;
	}
\end{procedure}
\end{figure}

The set $P$ maintains the pairs of states that result in a union step. To compute a witness in the case of insecurity, the store data structure keeps track of  
the justification for each union step. 
Every entry of \store consists of three pairs of the form $(s, t),(s', t'), (a, b)$. Such an entry is stored if a union is applied on $s$ and $t$ in a step, where $s$ is reached from $s'$ by an action (or the empty trace) $a$ and $t$ is reached from $t'$ by $b$. Since a union is only applied if $\findset(s) \neq \findset(t)$ and after the union $\findset(s) = \findset(t)$ holds, only at most one stored entry has $(s, t)$ as a first pair. 
If the union is performed on states with different observations, the system is insecure and the $\computewitness$ procedure computes a witness for insecurity, i.e., two 
runs that have the same $\purge_u$ value, but end
in states with different observations. 

To reference different values of the $\findset$ function, we
parameterize it with the number of unions, done during each iteration of the outer foreach-loop, 
i.e., for a fixed value of domain $u$. 
The function $\findset_i$,  
for $i\geq 0$, 
denotes the values of $\findset$ after the
$i$-th application of union during one iteration of the outer foreach-loop.  We call $i$ a union-number (of this iteration). 

Similarly, $P_i$ denotes the set of pairs in $P$ 
after the $i$-th application of union. Note that $P$ maintains this value until the line just before the next application of union. 
The set $P_{\leq i} = \bigcup_{j \leq i} P_j$ is the set of pairs inserted into $P$ up to the $i$-th union-number, including the pairs that are removed from $P$. 
Since $P_{\leq i}$ is a set of pairs of states, we will consider it as a binary relation on $\States$. 

Define for all $s, t \in \States$
\begin{equation*}
	s \eqiobs{u} t  \text{ iff } \obs_u(s) = \obs_u(t)
	\enspace.
\end{equation*}

To refer to stages in the construction of the 
equivalence relation,
define for all union-numbers $i$ and 
for all $s, t \in \States$
\begin{align*}
  s \eqifind{i} t & \text{ iff } \findset_i(s) = \findset_i(t) \\
  s \eqiscall{i} t & \text{ iff } s \eqifind{i} t \text{ and } s \cdot a
  \eqifind{i} t \cdot a \text{ for all } a \in \Actions 
\end{align*}
Note that both relations are equivalence relations on $S$. 
Note also that these relations are monotone in $i$, i.e., if $s
\eqifind{i} t$ then $s \eqifind{i+1} t$ and if $s \eqiscall{i} t$ then
$s \eqiscall{i+1} t$.

To show correctness of the algorithm, we will show that is sufficient to guarantee \scp and \ocp for the pairs of states collected in $P$. 
\begin{lemma}
  Fix an iteration of the outer foreach-loop, and let $i$ be a union-number of this iteration. 
Then $\eqifind{i}$ is the smallest equivalence relation on $\States$ that includes $P_{\leq i}$. 
\end{lemma}
\begin{proof}
By induction on $i$. For $i=0$ we have $P_i$ empty, so the smallest equivalence relation is the identity relation, which is 
also the relation corresponding to the initial value of $\findset_i$. At each union, 
the pairs inserted into $P$ are exactly those that are used by the union. Therefore, $\eqifind{i}$ is the reflexive, symmetric, transitive closure of $P_{\leq i}$. Hence, $\eqifind{i}$ is the smallest equivalence relation on $\States$ that includes $P_{\leq i}$. 
\end{proof}

First, we show that the witness produced by the algorithm is correct. 

\begin{lemma}
 \label{lem:pinsecty-alg}
 If the algorithm terminates with a witness $(\alpha, \beta)$, the analyzed system is P-insecure and we have 
 $\purge_u(\alpha) = \purge_u(\beta)$ and $\obs_u(s_0 \cdot \alpha) \neq \obs_u(s_0 \cdot \beta)$, where $u$ is the  agent chosen by the outer foreach-loop in the iteration where the compute-witness procedure is called. 
\end{lemma}
\begin{proof}
 The procedure compute-witness is called with two states as parameter having different observations for an agent $u$. 
 From these states, the compute-witness procedure constructs two paths back to $s_0$. From the stored values, it follows that these paths only differ in actions $a$ with $\dom(a) \not\nintrel u$. Therefore, we have $\purge_u(\alpha) = \purge_u(\beta)$. 
 Hence the constructed paths $\alpha$ and $\beta$ are a witness for the insecurity of the system. 
\end{proof}

The correctness of the algorithm follows from this Lemma, showing that the converse holds, too. 
\begin{lemma}
 \label{lem:psecty-alg}
 If the algorithm terminates with ``secure'', the analyzed system is \psec. 
\end{lemma}
\begin{proof}
	Consider one iteration of the outer foreach-loop, where some agent $u\in \Dom$ is chosen. 
 Let $m$ be the last union-number after this iteration of the outer foreach-loop.  
 We will show that the unwinding conditions \lrp, \scp and \ocp hold for $\eqifind{m}$. 
 It is clear that after the foreach-loop starting in line~\ref{alg:lrcond} the \lrp condition holds for $\eqifind{m}$, since for every $s \in \States$ and every $a \in \Actions$ with $\dom(a) \not\nintrel u$, we have $s \eqifind{m} s \cdot a$. 
 
 After each union step, it is checked that the observations are equal for the two merged states. 
 Assuming inductively that the observations are all equal on each of the two merged set of states, it follows that the observations are equal on the merged set. Therefore, for every $i \leq m$ we have ${\eqifind{i}} \subseteq {\eqiobs{u}}$. 
  Moreover, the relation $\eqifind{m}$ satisfies \ocp. 

 After this iteration of the outer foreach-loop,  $P_m$ is empty and therefore, we have for every $s, t \in \States$ that 
 if 
 $(s, t) \in P_{\leq m}$, then $s \eqifind{m} t$ and $s \cdot a \eqifind{m} t \cdot a$ for all actions $a \in \Actions$.
 This is guaranteed by the foreach-loop in line~\ref{alg:foreachloop} for every pair taken out of $P$ in the line above. 
 Therefore $P_{\leq m} \subseteq {\eqiscall{m}}$. 
 Since $\eqiscall{m}$ is an equivalence relation and since $\eqifind{m}$ is the smallest equivalence relation that contains $P_{\leq m}$, it follows that ${\eqifind{m}} \subseteq {\eqiscall{m}}$.
 Therefore, $\eqifind{m}$ satisfies \scp. 
\end{proof}

The following Lemma shows the correctness of the compute-witness procedure and gives a bound for its running time. 

\begin{lemma}
The procedure compute-witness computes a witness in $O(|\States| \cdot |\Actions|)$. 
\end{lemma}
\begin{proof}
 First, we will show that the graph induced by the stored values is a directed rooted tree. 
 For every stored entry $e$ of the form $e = ((s, t), (s', t'), (a, b))$ consider the projections 
 $\pi_0(e) = (s, t)$, $\pi_1(e) = (s', t')$ and $\pi_2(e) = (a, b)$. 
 For every union-number $i$, the graph of the stored values is $G_i = (V_i, E_i)$ with $V_i = \{\pi_0(e), \pi_1(e) \mid e \text{ is a stored entry up to the $i$-the union number}\}$ and $E_i = \{(\pi_0(e), \pi_1(e)) \mid e \text{ is a stored entry up to the $i$-the union number}\}$. 
 We will show by an induction on $i$, that the connected components of $G_i$ are directed rooted trees where all edges are oriented towards some root and all roots are of the form $(s,s)$ and that $P_{\leq i} \subseteq V_i \subseteq P_{\leq i} \cup \{(s, s) \mid s \in \States \}$ holds. 
 In the iterations of the loop for the LR-conditions, only edges of the form $((s, t), (s', s'))$ with $s \neq t$ are inserted. Therefore the resulting graph consists of directed rooted trees.  
 In the iterations of the loop for the SC-conditions, if an edge of the form $e = ((s, t), (s', t')$ is inserted into $G_{i}$, then 
 $\findset_{i}(s) \neq \findset_{i}(t)$ and therefore 
 $(s, t) \not\in P_{\leq i} \cup\{ (s, s) \mid s \in \States \}$
 and by induction hypothesis, $(s, t) \not\in V_{i}$. 
 Since $(s', t') \in P_{\leq i}$, the edge $e$ connects a new vertex with a vertex from $V_{i}$. Therefore, the resulting graph $G_{i+1}$ is again a directed rooted tree. 
 
 It follows that if the procedure compute-witness is called with some states $(s,t) \in P_{\leq i}$ then it finds a path to the unique root $(s', s')$ of the connected component of $(s, t)$ within $O(|S|)$. 
 A shortest path from $s'$ to $s_0$ can be found in $O(|S| \cdot |A|)$. 
\end{proof}

The running time of the whole algorithm can be analyzed as follows. 
Since the union is only applied on states in different sets, 
for each $u\in D$, the total  
number of unions is bounded by $|\States|$. Note, that the insertion into $P$ is only combined with a union step, therefore the sets of pairs inserted into $P$ during the whole run of the algorithm is bounded by $|\States|$, too. 
Also, the number of stored triple of pairs is the same as the number of unions, since for every pair $(s, t)$ the union is only applied at most once. The stored values provide a function from the first to the second and third pair. 
The union-find operations can be implemented so as to 
have an amortized cost of $ \alpha(|S|)$, where $\alpha$ is the 
very slow growing (effectively constant for practical purposes) ``inverse'' of Ackermann's function
\cite{Tarjan75}. Thus, the running time of this algorithm is bounded by 
$O(|D| \cdot |A| \cdot |S| \cdot \alpha(|S|))$.

\subsection{Complexity of  \ipsecty} \label{sec:alg-ip}

An approach similar to that for \psecty\ works in the case of \ipsecty, based
on the unwinding characterization of \ipsecty. The 
algorithm is given in Figure~\ref{fig:alg-ipsecty}.

\begin{figure}
\begin{algorithm}[H]
	\caption{Decide \ipsecty}
 \ForEach{ $u \in \Dom$} 
  {
		\ForEach{ $v \in \Dom$ with $v \not\nintrel u$ }
		{
			\tcc{ create a new partition } 
			\ForEach{ $ s \in \States$ } 
			{ 
				\makeset($s$) \;
			}
			let $P$ be an empty list \;
			let $\store$ be empty \;
			\tcc{ apply LR conditions } 
			\ForEach{ $s \in \States$ } 
			{ 
				\ForEach{ $a \in \Actions$ with $\dom(a) = v$ } 
				{ 
					\If{ $\findset(s) \neq \findset( s \cdot a )$ }
					{
						add $((s\cdot a, s), (s, s), (a, \epsilon))$ to \store \;
						insert $(s\cdot a, s)$ into the list $P$ \;
						$\union(s\cdot a, s)$ \; 
						\If{ $\obs_u(s \cdot a) \neq \obs_u(s)$ }
						{
							\Return $\computewitness(s\cdot a, s, \epsilon, \epsilon)$ \;
						}
					}
				}
			}
			\tcc{ apply SC conditions } 
				\While{ $P \neq \emptyset$ } 
				{ 
					take a pair $(s, t)$ out of $P$ \;
					\ForEach{ $a \in \Actions$ with $v \not\nintrel \dom(a)$ }
					{ 
						\If{ $\findset(s \cdot a) \neq \findset(t \cdot a)$ }
						{	
							add $((s\cdot a, t\cdot a), (s, t), (a, a))$ to \store \;
							insert $(s \cdot a, t\cdot a)$ into the list $P$ \;
							$\union(s\cdot a, t\cdot a)$ \;
							\If{ $\obs_u(s \cdot a) \neq \obs_u(t \cdot a)$ }
							{
								\Return $\computewitness(s\cdot a, t\cdot a, \epsilon, \epsilon )$ \;
							}
						}
					}
				}
			}
		}
\Return{``secure''}  
\end{algorithm}
\caption{An algorithm for \ipsecty.\label{fig:alg-ipsecty}}
\end{figure}

The argument for correctness  is similar to that for 
the algorithm for \psecty: the algorithm  computes the minimal equivalence relations
satisfying \scip  and \lrip, and checks that these satisfy \ocip. 
The argument for the SC conditions is similar to that of Lemma~\ref{lem:pinsecty-alg} and \ref{lem:psecty-alg}. 
The running time of the algorithm for checking \ipsecty\ is 
$O(|D|^2 \cdot |A| \cdot |S| \cdot \alpha(|S|))$.

\subsection{Complexity of \tasecty}   \label{sec:alg-ta}

\begin{figure}
\begin{algorithm}[H]
	\caption{Decide \tasecty}
	Run the Algorithm for IP-security: if it returns "secure", then  \\
	\ForEach{ $u \in \Dom$} 
  {
		\ForEach{ $v \in \Dom$ }
		{ 
			\ForEach{ $w \in \Dom$ with $w \not\nintrel v$ and $v \not\nintrel w$ and $w \not\nintrel u$ }
			{
				\tcc{ create a new partition } 
				\ForEach{ $ s \in \States$ } 
				{ 
					\makeset($s$) \;
				}
				let $P$ be an empty list \;
				let $\store$ be empty \;
				\tcc{ apply LR conditions } 
				\ForEach{ $s \in \States$ } 
				{ 
					\ForEach{ $a \in \Actions$ with $\dom(a) = v$ } 
					{
						\ForEach{ $b \in \Actions$ with $\dom(b) = w$ }
						{
							\If{ $\findset(s \cdot ab) \neq \findset( s \cdot ba )$ }
							{
								add $((s\cdot ab, s\cdot ba), (s, s), (ab, ba))$ to \store \;
								insert $(s\cdot ab, s\cdot ba)$ into the list $P$ \;
								$\union(s\cdot ab, s\cdot ba)$ \; 
								\If{ $\obs_u(s \cdot ab) \neq \obs_u(s\cdot ba)$ }
								{
									\Return $\computewitness(s\cdot ab, s\cdot ba, \epsilon, \epsilon)$ \;
								}
							}
						}
					}
				}
				\tcc{ apply SC conditions } 
				\While{ $P \neq \emptyset$ } 
				{ 
					take a pair $(s, t)$ out of $P$ \;
					\ForEach{ $a \in \Actions$ with $v \not\nintrel \dom(a)$ or $w \not\nintrel \dom(a)$ }
					{ 
						\If{ $\findset(s \cdot a) \neq \findset(t \cdot a)$ }
						{
							add $((s\cdot a, t\cdot a), (s, t), (a, a))$ to \store \;
							insert $(s \cdot a, t\cdot a)$ into the list $P$ \;
							$\union(s\cdot a, t\cdot a)$ \;
							\If{ $\obs_u(s \cdot a) \neq \obs_u(t \cdot a)$ }
							{
								\Return $\computewitness(s\cdot a, t\cdot a, \epsilon, \epsilon )$ \;
							}
						}
					}
				}
			}
		}
	}
\Return{``secure''}  
\end{algorithm}
\caption{An algorithm for \tasecty.\label{fig:alg-tasecty}}
\end{figure}

The argument for correctness  of the algorithm of Figure~\ref{fig:alg-tasecty} is similar to that for 
the algorithm for \psecty: the algorithm  computes the minimal equivalence relations
satisfying \scta  and \lrta, and checks that these satisfy \octa. 
The argument for the SC conditions is similar to that of Lemma~\ref{lem:pinsecty-alg} and \ref{lem:psecty-alg}. 
The running time is $O(|D|^3 \cdot |A| \cdot |S| \cdot \alpha(|S|))$.

\subsection{Space Complexity and Symbolic Implementation for \psecty, \ipsecty\ and \tasecty}

The  algorithms presented above for \psecty, \ipsecty\ and \tasecty\ 
all require space linear in the size of the number of states of the system. Due the ``state space explosion" problem, 
i.e.,  the fact that the number of states of a system grows exponentially with the number of 
state variables, the number of states may be very large in practical examples. 
For each of these notions of security, it is possible to trade off space for time, at the cost of introducing nondeterminism. Instead of computing an explicit representation of the minimal unwinding relation $\unwind$ of 
the appropriate type, we search through a graph in which the vertices are pairs of states $(s,t)$ such that 
$s\unwind t$, and in which edges correspond to one of the unwinding rules of type LR or SC. 
The search begins at a vertex $(s,s)$, where $s$ is reachable state, and 
terminates and declares ``insecure" if a pair $(s,t)$ is reached such that $\obs_u(s) \neq \obs_u(t)$.   
The search can be conducted using nondeterminism and terminated at depth $|S|^2$ 
with a declaration of ``secure" in the case that no such pair is found. 
This gives a co-NLOGSPACE =NLOGSPACE  algorithm. 
Since graph search trivially reduces to these problems, verification of all three 
security notions is complete for nondeterministic logarithmic space.

Both the algorithms given above and the nondeterministic logarithmic space
approach rely on explicit representation of states. 
We note that, in practice, an effective approach to deciding the three properties 
may be to instead use {\em symbolic} representations of states and relations to 
represent the fixpoint computation for the minimal relation on reachable states 
satisfying the SC and LR type rules, 
following techniques well known from the model checking area \cite{BCMDH92}, 
and then to intersect with a symbolic representation of the complement of the OC type rule 
and check for non-emptiness. The performance of this approach is unpredictable in general, 
so comparison with the algorithms above is a matter for experimental research. 

\subsection{TO-security}

We now prove that \tosecty\ is undecidable. The proof also shows that the source of the undecidability does not lie in using 
complex policies, in fact the problem remains undecidable for a very simple, small policy.

\medskip

\begin{theorem} 
\label{thm:to-undecidable}
It is undecidable 
whether $M$ is \tosec\ with respect to $\nintrel$,
even for a fixed policy containing $4$ domains.
\end{theorem} 

\medskip

\begin{proof} 
We prove the undecidability of \tosecty\ by a reduction from the Post
Correspondence Problem \cite{Post}.  An instance of this problem
consists of a pair of sequences  ${\cal U} = U_1, \ldots, U_n$ and
${\cal W} = W_1, \ldots, W_n$ of words over an alphabet $\Sigma$ with at least
two letters.  The problem PCP is the set of such pairs  $({\cal
U,W})$ such that there exists a sequence of indices $i_1, \ldots, i_k$ with
$1\leq i_j \leq n$ for each $j=1,\ldots, k$, such that $U_{i_1} \ldots
U_{i_k} = W_{i_1} \ldots W_{i_k}$. We encode an instance of this
problem as a machine $M({\cal U,W})$ for the (intransitive) policy for
agents $A,B,C,D$ given by $A \nintrel C$, $A\nintrel D$, $B\nintrel C$
and $C\nintrel D$, such that  $({\cal U,W})\in PCP$ iff $M({\cal U,W})$ 
is not \tosec\ with respect to $\nintrel$. 

Intuitively, in the machine $M({\cal U,W})$, agent $A$ guesses a word
over $\Sigma$, and agent $B$ chooses whether this word is to be compared to a
sequence of $U_i$ or $W_i$, and guesses a sequence of indices used to
make the comparison. 
Agent $C$ observes the indices guessed by $B$, and guesses when the word being constructed is complete. 
Agent $D$ observes nothing until $C$ declares the end of the construction, and 
then observes  whether the word guessed by $A$ does in fact correspond to the sequence of
indices guessed by $B$. 
The definition of \tosec\ will be guaranteed to
hold with respect to agents $A,B$ and $C$, so the determination as to whether
$M({\cal U,W})$ is \tosec\ depends on how the observations of agent
$D$ relate to the actions and observations of $A$ and $C$.
More precisely, $M({\cal U,W})$ has 
\be 
\item states of the form $(p,V,i,x)$, where 
\be

\item $p\in \{U,U',W\}$ indicates whether the sequence of letters
  guessed by $A$ is to be compared with a sequence of $U_i$ (when
  $p\in \{U,U'\}$) or as a sequence of $W_i$ (when $p=W$).

\item $V$ is either a word over $\Sigma$ which is a prefix (possibly the empty word $\epsilon$) 
of one of the $U_i$ or $W_i$, or $\top$.  Intuitively, this indicates a part of
the word guessed by $A$ that will be compared to an index guessed by $B$. 
The case of $\top$ represents that an inconsistency has been detected.%
\footnote{Throughout, we use $\bot$ to represent undetermined information and 
$\top$ to represent inconsistency.}

\item $i\in \{0, \ldots, n\}$ is either $0$ (no activity so far) or 
the last index guessed by $B$, 

\item $x\in \{0,1\}$ is used to represent the state of the computation, 
with $0$ meaning ongoing and $1$ meaning complete. 

\ee

\item initial state $(U,\epsilon,0,0)$, 

\item actions 
\be 
\item of $A$: an action $a$ for each $a\in \Sigma$, corresponding to 
guessing the letter $a$
\item of $B$: an action $w$ (corresponding to the selection of $W$) 
 plus an action $g_i$ for each $i\in \{1,\ldots, n\}$ (corresponding to 
a guess of the index $i$) 
\item of $C$: an action $end$ 
\item of $D$: none 

\ee

\ee The transition function is defined as follows.  For all actions
$b$ and states $s= (p,V,i,x)$, if $x=1$ then we will have
$\step(s,b)=s$, i.e., once the computation has terminated, no action
changes the state. We therefore confine the definitions below to the
case $x=0$.  We make use of two functions $G\colon \{U,U',W\} \rightarrow
\{U,W\}$ defined by $G(U)=G(U')=U$ and $G(W)=W$, and $F\colon \{U,U',W\}
\rightarrow \{U',W\}$ defined by $F(U)=F(U') = U'$ and $F(W) = W$.

In a state $(p,V,i,x)$, the value $G(p)$ captures the choice of ${\cal
U}$ or ${\cal W}$ with which to compare the word being generated by
$A$. Intuitively, both $p=U$ and $p=U'$ represent that the word being
processed is to be compared with the ${\cal U}$. This the default, as
indicated in the initial state. The reason for including $U'$ is that
agent $B$ is given an opportunity to switch the system to comparing
with ${\cal W}$ only in the first step of a run. If it does not act,
then the choice remains with ${\cal U}$ for the remainder of the run.  

In the case of action $w$, we define $\step((p,V,i,0),w) = (W,V,i,0)$ if
$p= U$ and $\step((p,V,i,0),w) = (p,V,i,0)$ otherwise.
This says that $w$ switches the choice of comparison to ${\cal W}$. 
That the choice can be made only if $w$ is the initial action of a
run is captured by defining all other actions $b\neq w$ so that if 
$\step((p,V,i,0),b) = (p',V',i',x')$ then $p' = F(p)$. 

For the actions $a$ of $\Actions$, we define $\step((p,V,i,x),a) = 
(F(p),V',i,x)$, where $V' = V\cdot a$ if $V\cdot a$ is a prefix of
$G(p)_j$ for some $j$, and $V' = \top$ otherwise. Intuitively, $V$ is
used to collect a fragment of the sequence being generated by $A$ for
comparison with the $G(p)_j$.  We accumulate the fragment while it is
a prefix of such a string, and as soon as this is no longer the case
we flag the inconsistency.

For the actions $g_j$ of $B$, we define 
 $\step((p,V,i,0),g_j) = (F(p),V',j,0)$, where 
\be 
\item if $G(p)_j= V$ then $V' = \epsilon$, and  
\item if $G(p)_j\neq  V$ then $V' = \top$.
\ee 
Intuitively, this captures that the effect of the action $g_j$ is to 
compare $G(p)_j$ with the current fragment of the string being generated 
by $A$. If they are equal, we reset $V$ to $\epsilon$ in order to check the next fragment. 
Otherwise, we flag the inconsistency. 

For the action $end$ of $C$, we define 
 $\step((p,V,i,0),end) = (F(p),V',i,1)$, where 
\be 
\item if $V = \epsilon$ then $V'=\epsilon$, and  
\item if $V \neq \epsilon$ then $V'=\top$.  
\ee 
Intuitively, this action checks that the end is declared at a time 
when there is no fragment currently being processed, and flags an 
inconsistency otherwise.

The observations are now defined as follows. The observations of $A$
and $B$ are trivial: $\obs_A(s) = \obs_B(s) = \bot$ for all states
$s$.  For $C$, we define 
\begin{equation*}
\obs_C((p,V,i,x)) =
\begin{cases}
i &  \text{if $V = \epsilon$} \\
\top & \text{if $V=\top$} \\ 
\bot & \text{otherwise.}
\end{cases}
\end{equation*}
Note that this means
that for the initial state $s_0$ we have $\obs_C(s_0) = 0$.
Intuitively, since $i$ records the last (successful) guess of index
for a fragment of the word being generated by $A$, we have that $C$
becomes aware of a guess whenever it is correct, and can see from its
observation $\bot$ that a further fragment is in the process of being
constructed.  

For $D$ we define
\begin{equation*}
\obs_D((p,V,i,x))
 =
\begin{cases} 
\bot & \text{when $x=0$,} \\
G(p) & \text{when $x=1$, and $V = \epsilon$ and $i\neq 0$} \\
\top & \text{when $x= 1$ and either $V \neq \epsilon$ or $i=0$.}
\end{cases}
\end{equation*}
Intuitively, this means that $D$ observes $\bot$  until $C$ declares the 
end of the string, and learns whether ${\cal U}$ or ${\cal V}$ fragments
were being checked when a decomposition has been successfully guessed. 
Otherwise, it learns that the guesses do not match. 
This completes the definition of the system $M({\cal U}, {\cal W})$. 

\ifallproofs

We leave it to the reader to check that the conditions for \tosecty\
are satisfied in $M({\cal U}, {\cal W})$ for the agents $u=A,B,C$. For
$u=D$, we claim that the definition is violated iff there exists a
sequence of indices $i_1, \ldots i_k$ such that $U_{i_1}\ldots U_{i_k}
= W_{i_1}\ldots W_{i_k}$.  For, suppose that such a sequence exists.
Define
\begin{equation}\label{eq:pcpo}
\alpha =
U_{i_1}g_{i_1} \ldots g_{i{k-1}}U_{i_k}g_{i_k}end
\end{equation} 
and 
\begin{equation} \label{eq:pcpt}
\alpha' = w
W_{i_1}g_{i_1} \ldots g_{i{k-1}}W_{i_k}g_{i_k}end.
\end{equation} 
 Then we have
$\purge_D(\alpha) = U_{i_1}\ldots U_{i_k}end= W_{i_1}\ldots W_{i_k}end
= \purge_D(\alpha')$.  Moreover, $v\nintrel D$ and $v\neq D$ iff 
$v=A$ or $v = C$. 
If $U_{i_1}\ldots U_{i_k}$ is the sequence of letters $a_1a_2\ldots a_m$, 
then $\tview_A(\alpha) = \bot a_1 \bot a_2\ldots \bot a_m  = \tview_A(\alpha')$. 
Also, we have $\tview_C(\alpha) = 0\bot i_1\bot i_2 \ldots i_k \bot end = \tview_C(\alpha)$. 
 However $s_0\cdot \alpha = (U',\epsilon, i_k,1)$ and
$s_0\cdot \alpha' = (W,\epsilon, i_k,1)$, so $\obs_D(s_0\cdot \alpha)
= U \neq W = \obs_D(s_0\cdot \alpha')$. Thus  the condition for 
\tosecty\ is violated. 

Conversely, suppose that the condition for \tosecty\ is violated. 
As noted above, the violation can only occur for agent $D$.  Let the
witness be the sequences of actions $\alpha, \alpha'$.  Then we have
$\purge_D(\alpha) = \purge_D(\alpha')$ 
and 
$\tview_A(\alpha) = \tview_A(\alpha')$ 
and
$\tview_C(\alpha) = \tview_C(\alpha')$ and $\obs_D(s_0\cdot \alpha)
\neq \obs_D(s_0\cdot \alpha')$.

The action $end$ must occur in at least one of these sequences, else
we have $\obs_D(s_0\cdot \alpha) = \obs_D(s_0\cdot \alpha') = \bot$,
since only the action $end$ can set $x=1$.  Since 
$\tview_C(\alpha) = \tview_C(\alpha')$, and the action $end$ of $C$ is recorded in its 
possibly transmitted observations, it follows that $end$ in
fact occurs in both sequences.  Since no action changes the state
after the occurrence of $end$, we may assume without loss of
generality that $end$ is the final action in both sequences.
(Deleting any subsequent actions preserves the observations. 
Any actions of 
$A$ or 
$C$ must occur in both, so their deletion 
maintains equality of $\purge_D$ 
and $\tview_A$ and 
and $\tview_C$.
All other actions can be deleted without change to $\purge_D$, 
$\tview_A$ 
or $\tview_C$.)  
Similarly, we may also assume that any occurrence of $w$ must be as
the initial action in the sequence, since otherwise this action
changes nothing. 

Since both sequences contain $end$ we do not have $\obs_D(s_0\cdot \alpha) =\bot$. 
Let $s_0\cdot \alpha = (p,V,i,1)$ and suppose that
$\obs_D(s_0\cdot \alpha) = \top$. 
Then 
we have $V\neq \epsilon$ or $i=0$. 
There are several possibilities. We show that 
each implies $\obs_D(s_0\cdot \alpha') =\top$, yielding a contradiction with 
$\obs_D(s_0\cdot \alpha) \neq \obs_D(s_0\cdot \alpha')$. The different cases are:
\begin{enumerate}
\item If $i=0$ then $\alpha$ contains no action $g_j$, so $\tview_C(\alpha)$
  is $0 end$ or $0\bot end$.  Hence $\tview_C(\alpha')$ also has this form,
  from which it follows that $\alpha'$ also contains no action $g_j$, so
  $\obs_D(s_0\cdot \alpha') =\top$.
\item If $V = \top$ then at some stage of the computation either the current
  fragment ceased to be a prefix of any $G(p)_j$, or $B$ guessed an index not
  matching the current fragment.  In either case, $\top$ must occur in
  $\tview_C(\alpha)$, hence in $\tview_C(\alpha')$ also. This implies, by the
  definitions of $\obs_C$ and $\obs_D$ and the fact that once $V$ becomes
  $\top$ it remains $\top$, that $\obs_D(s_0\cdot \alpha') =\top$.
\item Alternatively, if $V \in \Sigma^*\setminus \epsilon$, then
  $\tview_C(\alpha)$ ends with $\bot end$, hence also $\tview_C(\alpha')$ ends
  with $\bot end$. This means that in $\alpha'$ a fragment was under
  construction when $end$ occurred, hence if $s_0\cdot \alpha' = (p',V',i',1)$
  then we have $V'\neq \epsilon$, and hence $\obs_D(s_0\cdot \alpha') =\top$
\end{enumerate}

It therefore follows that $\obs_D(s_0\cdot \alpha)$ and $\obs_D(s_0\cdot \alpha')$
take (since they differ) the values $U$ and $W$. 
The sequences of the $g_j$ in $\alpha$ must be the same, since 
no errors were detected in either computation, 
so for each occurrence $g_j$ we have that $j$ appears in $\tview_C$. 
The last action before $end$ in these sequences 
must be an action $g_j$, else we end in a state $(p,V,i,1)$ with $V\neq \epsilon$
or  $i=0$. It follows that $\alpha$ and $\alpha'$ have the 
forms in equations~(\ref{eq:pcpo}) and~(\ref{eq:pcpt}), and that 
$({\cal U,W})\in PCP$. 
\fi
\end{proof}\

We note that the undecidability result for \tosecty\ implies that
there are no simple unwinding conditions that are complete for this
notion of security.  In particular, any first-order set of conditions
on a collection of binary relations on states can be checked in 
PTIME, hence cannot be both sound and complete. 

\subsection{ITO-security} 

A proof for the undecidability of \itosecty\ could be given that is similar to that for \tosecty. 
However, the result can also be obtained by noting that the similarity of the two definitions allows
for a reduction from \tosecty\  to \itosecty. The details of the reduction are given in the appendix. 
The following result is then immediate from Theorem~\ref{thm:to-undecidable}.

\begin{theorem} 
It is undecidable to determine whether $M$ is \itosec\ with respect to $\nintrel$, even for a fixed policy with 
4 domains. 
\end{theorem}

\section{Related Work}\label{sec:related} 

The notion of \emph{noninterference} was first proposed by Goguen and Meseguer \cite{GogMes}.  
Early work in this area was motivated by multi-level secure systems,
and dealt with 
partially ordered (hence transitive) information flow policies.  
The simplest of these is the
two-domain policy with domains $L$ and $H$ and  $L \nintrel H$, but not $H \nintrel L$. 
Much of the literature is confined to this simple policy. 
Even with this restriction, there exists a large set of proposed
definitions of noninterference \cite{sutherland86,WJ90,mccullough88,FG01,Ryan01}.  
These definitions generally agree when applied to deterministic systems,
and the differences relate to how the definitions should behave on nondeterministic
systems. In addition to state-observed systems model used in the present paper,
a variety of systems models have been considered, including  action-observed
systems,  trace semantics, and process algebraic semantics (both CCS and CSP flavours).
A number of works have sought to classify the definitions when
formulated in a very general process algebraic setting
\cite{FG01}, as well as to establish formal relations between
definitions cast in different semantic models \cite{MZ10}.

The main point of overlap  of our work with this literature is to consider 
how our  results concerning \psecty, when applied to transitive policies, 
relate to other algorithmic verification approaches in the literature for such policies. 
Our approach here is similar to other work in the area. 
In particular, the idea of running two copies of the system in parallel,  
in order to compare two different runs, has been used before \cite{BartheDR04}. 
Other approaches have been developed for automated verification of 
noninterference  based on process algebraic bisimulation techniques \cite{FG01,FG96}. 
Mantel \cite{Mantel-thesis} has characterised many of the existing definitions of noninterference 
as compositions of a set of Basic Security Properties. The complexity of verifying these
basic properties  has been studied \cite{DSouzaRS05}. 
A few works have considered richer systems models than
finite state systems, e.g. pushdown systems \cite{DSouzaHKRS08}.
We note that our results in this paper, and in much of the literature, is concerned 
with asynchronous systems in which agents are unaware of the passage of time. 
Some of the literature deals with synchronous systems, where a similar spectrum of 
definitions of noninterference exists for nondeterminisitic systems. 
Some recent work has  investigated verification of synchronous notions of noninterference \cite{KB06,CMZ10}.

Some work on development of tools based on decidable cases of
such definitions of noninterference has been performed. Focardi et al. describe a 
tool based on process algebraic techniques \cite{FGP95}. 
Whalen et al. \cite{Whalen10} present an approach to model checking
noninterference that is in use at Rockwell Collins for
verification of MILS systems.  
Their approach is a mix of model checking and static analysis, in which a model checker is
used to search through an enriched version of the model in which labels
computed by static analysis are associated to systems components.
They formally prove it to be sound  with respect to a definition of noninterference
from work by Greve, Wilding and van Fleet \cite{GWV03}.  While they discuss examples
requiring intransitive policies, they leave details of this for future work.

Since we have confined ourselves in this paper to deterministic systems, 
but focus on richer intransitive policies, much of the work discussed above, 
which is confined to transitive policies and nondeterministic systems, 
is orthogonal to our concerns. 
Algorithmic verification of intransitive noninterference has had less attention in the literature.
After the work of Rushby \cite{rushby92}, \ipsecty\ has generally been taken to be
the definition studied.

Pinsky
\cite{pinsky95} presents a PTIME procedure for deciding \ipsecty\  that, in effect, generates a
relation that is claimed to satisfy  Rushby's unwinding conditions for transitive
noninterference just when the system is secure. However, in fact the
relation may fail to satisfy the Output Consistency condition,
so this claim is incorrect.
(Pinsky's argument supporting the claim that the relation
satisfies  Output Consistency, in the corollary to Theorem 2,
states that $SA(basis_\pi(z),\alpha)$ is a subset of
$view(state\_action(z,\alpha))$. This  is correct for transitive policies,
but could be false for intransitive policies.)
That such an approach cannot work for \ipsecty\ also follows from
results in \cite{meyden2007}, where it is shown that Rushby's unwinding conditions
are sound also for \tasecty,  which is a stronger notion than \ipsecty.
Moreover, an example in \cite{meyden2007} shows that  a system may be
\tasec, but no Rushby unwinding exists on the system (although one
will exist on the infinite state unfolded system, when the system satisfies
\tasecty).
Thus, an approach based on finding a Rushby unwinding, on the system as given, will also fail to be 
complete for the notion of TA-security. 

Hadj-Alouane et al. \cite{HLFMY05} also present a decision
procedure for \ipsecty, but it has  complexity $O(2^{|\States |.2^{|D|}})$, which is less efficient than our procedure by two exponentials. 

Roscoe and Goldsmith \cite{RoscoeG99} have presented a critique of \ipsecty\ 
(arguing that it is too liberal in the information flows it permits), and have proposed two alternate 
definitions cast in the process algebra CSP, based on what they call {\em lazy} and 
{\em mixed} abstraction operators. It has been shown by van der Meyden \cite{meyden-compini} 
that the definition based on lazy abstraction corresponds to \psecty, and the 
version based on mixed abstraction corresponds to 
\itosecty. 
Roscoe and Goldsmith give an informal discussion, 
without precise complexity results  or proof, 
of algorithms for deciding ``the generalised noninterference condition". 
Based on van der Meyden's characterization of the definition based on the mixed abstraction 
as 
as corresponding to \itosecty, we obtain that the definition based on mixed abstraction is
undecidable.
It therefore seems that
their comments should be interpreted as concerned (like most of the preceding content in their paper) just with 
the lazy abstraction based definition, and hence comparable to our
PTIME result for \psecty.  

Unwinding was introduced in \cite{goguen84} and given a crisp 
presentation in \cite{rushby92}, for both transitive and intransitive 
security policies. In particular, Rushy shows that a notion of 
unwinding for transitive policies is both sound (if an unwinding exists, 
then the system is secure) and complete (if the system is secure, then there
exists an unwinding). He also defines a  notion of {\em weak unwinding} tailored to 
intransitive security policies, and proves its soundness for IP-security, but not completeness. 
The reason for this gap was identified in \cite{meyden2007}: weak unwinding 
is also sound for the strictly stronger notion of  TA-security, so cannot be 
complete for IP-security. 

In a context of systems with local policies, sound and complete unwindings are given in \cite{ESW13}. The unwindings for \ipsecty\ in this work are a special cases of the unwindings given there. 
\newcommand{\uf}{\mathit{uf}}

Van der Meyden also shows that 
weak unwinding is complete for TA-security in the sense that if a system $M$ is TA-secure, 
then there exists a weak unwinding on the system $\uf(M)$ obtained by {\em unfolding} $M$. 
The  system $\uf(M)$ is essentially the same as $M$, but with a new component added to 
states that records the complete history of the system. The system $\uf(M)$ is {\em bisimilar} to
$M$ in a appropriate sense of bisimilarity. It is generally held that two bisimilar systems 
are equivalent with respect to all properties of interest. However, somewhat unusually, 
the existence of a weak unwinding is not preserved by bisimilarity. This means that,  
while unfolding and then searching for a weak unwinding is a complete proof technique for 
TA-security, it falls short of providing a set of unwinding conditions on the system {\em as given}. 
This is significant, in that the existence of a first-order expressible set of such conditions 
would also yield a  decision procedure for TA-security of finite state systems. (The procedure would 
be of at most nondeterministic polynomial-time complexity, involving guessing an unwinding and 
verifying its properties.)  By contrast, the process of unfolding turns a finite state system 
into an infinite-state system, so a proof technique that relies upon it does not trivially 
yield a decision procedure. Indeed, the decidability of TA-security was an open problem
until the present work. 

In practice, definitions of the kind we have studied are very liberal in the information flows that they permit: when a (source) domain acts, 
everything that it knows (in some sense of knowledge) may be transmitted to any domain with which the source is permitted to interfere. 
In practice, one generally wants to limit the information that flows from one domain to another to be just a subset of the information available to the source. 
A framework for such policies that generalizes the definitions studied in the present paper 
has been developed by van der Meyden and Chong \cite{chong-meyden-tark09,chong-meyden-law12}.

Approaches to stating policies expressing such limitations have
also 
 been developed in the context of language-based approaches to security, where they are generally supported by means of sound but incomplete static analyses \cite{Mantel:Sands:APLAS04,Sabelfeld:Sands:CSFW05,CM04,BanerjeeNR08,MyersSZ04}. 
In the existing work, the policy is generally taken to be $L \nintrel H$ with exceptions allowed to $H \not \nintrel L$, or more generally, 
a partial order with exceptions. The system is given by a single, typically deterministic program, and the focus is on relating initial  values of input variables to 
final values of output variables, rather than on what can be deduced from ongoing observations in the state machine approach
we have considered here.

\section{Conclusion}\label{sec:concl} 

In this paper, we have determined the computational complexity of verifying whether a finite-state system satisfies an intransitive noninterference security property. The polynomial-time upper bounds build on new characterizations 
and unwindings for two of the notions 
of noninterference dealt with. 
They also allow counterexamples  (which can be used to improve the system in question) to be found when the system is insecure.

We have considered only deterministic systems:  there have been several proposals 
 to define intransitive noninterference in nondeterministic systems
\cite{young94,Mantel01,RoscoeG99,oheimb04,EMZ12}, the issue of complexity of these definitions remains open. 

It would also be desirable to investigate algorithms and complexity for information flow policies 
of the richer types studied in the literature on programming languages approaches to declassification, in order to 
obtain sound and complete approaches for such specifications.  Since intransitive noninterference policies provide a format for 
specifying architectural structure of a system, it would be interesting to  combine the strengths of the programming languages perspective
and the state machine model approach we have followed in this paper. 

{\bf Acknowledgments:} 
Work  supported by Australian Research Council Discovery grant DP1097203.
The second author thanks  the Courant Institute, New York University, 
and the Computer Science Department, Stanford University 
for hosting
sabbatical visits during which this research was conducted in part.

\bibliographystyle{IEEEtran}
\bibliography{refs}

\begin{thebibliography}{10}
\providecommand{\url}[1]{#1}
\csname url@samestyle\endcsname
\providecommand{\newblock}{\relax}
\providecommand{\bibinfo}[2]{#2}
\providecommand{\BIBentrySTDinterwordspacing}{\spaceskip=0pt\relax}
\providecommand{\BIBentryALTinterwordstretchfactor}{4}
\providecommand{\BIBentryALTinterwordspacing}{\spaceskip=\fontdimen2\font plus
\BIBentryALTinterwordstretchfactor\fontdimen3\font minus
  \fontdimen4\font\relax}
\providecommand{\BIBforeignlanguage}[2]{{%
\expandafter\ifx\csname l@#1\endcsname\relax
\typeout{** WARNING: IEEEtran.bst: No hyphenation pattern has been}%
\typeout{** loaded for the language `#1'. Using the pattern for}%
\typeout{** the default language instead.}%
\else
\language=\csname l@#1\endcsname
\fi
#2}}
\providecommand{\BIBdecl}{\relax}
\BIBdecl

\bibitem{EMSW11}
S.~Eggert, R.~van~der Meyden, H.~Schnoor, and T.~Wilke, ``The complexity of
  intransitive noninterference,'' in \emph{IEEE Symposium on Security and
  Privacy}.\hskip 1em plus 0.5em minus 0.4em\relax IEEE Computer Society, 2011,
  pp. 196--211.

\bibitem{GogMes}
J.~Goguen and J.~Meseguer, ``Security policies and security models,'' in
  \emph{Proc. IEEE Symp. on Security and Privacy}, Oakland, 1982, pp. 11--20.

\bibitem{HY87}
J.~Haigh and W.~Young, ``Extending the noninterference version of {MLS for
  SAT},'' \emph{IEEE Trans. on Software Engineering}, vol. SE-13, no.~2, pp.
  141--150, Feb 1987.

\bibitem{rushby92}
\BIBentryALTinterwordspacing
J.~Rushby, ``Noninterference, transitivity, and channel-control security
  policies,'' SRI International, Tech. Rep. {CSL-92-02}, Dec 1992. [Online].
  Available: \url{http://www.csl.sri.com/papers/csl-92-2/}
\BIBentrySTDinterwordspacing

\bibitem{meyden2007}
R.~van~der Meyden, ``What, indeed, is intransitive noninterference?'' in
  \emph{European Symposium On Research In Computer Security (ESORICS)}, ser.
  Lecture Notes in Computer Science, J.~Biskup and J.~Lopez, Eds., vol.
  4734.\hskip 1em plus 0.5em minus 0.4em\relax Springer, 2007, pp. 235--250.

\bibitem{goguen84}
J.~Goguen and J.~Meseguer, ``Unwinding and inference control,'' in \emph{IEEE
  Symp. on Security and Privacy}, 1984.

\bibitem{Denning}
D.~E. Denning, ``A lattice model of secure information flow,''
  \emph{Communications of the ACM}, vol.~19, no.~5, pp. 263--243, 1976.

\bibitem{BDRF08}
C.~Boettcher, R.~DeLong, J.~Rushby, and W.~Sifre, ``The {MILS} component
  integration approach to secure information sharing,'' in \emph{Proc. 27th
  IEEE/AIAA Digital Avionics Systems Conference}, Oct. 2008, pp.
  1.C.2--1--1.C.2--14.

\bibitem{MZ10}
R.~van~der Meyden and C.~Zhang, ``A comparison of semantic models for
  noninterference,'' \emph{Theoretical Computer Science}, vol. 411, no.~47, pp.
  4123--4147, Oct. 2010.

\bibitem{RoscoeG99}
A.~W. Roscoe and M.~H. Goldsmith, ``What is intransitive noninterference?'' in
  \emph{IEEE Computer Security Foundations Workshop}, 1999, pp. 228--238.

\bibitem{Tarjan75}
R.~E. Tarjan, ``Efficiency of a good but not linear set union algorithm,''
  \emph{J. ACM}, vol.~22, no.~2, pp. 215--225, 1975.

\bibitem{BCMDH92}
J.~Burch, E.~Clarke, K.~McMillan, D.~Dill, and L.~Hwang, ``Symbolic model
  checking: {10$^{20}$} states and beyond,'' \emph{Information and
  Computation}, vol.~98, no.~2, pp. 142--170, 1992.

\bibitem{Post}
E.~Post, ``A variant of a recursively unsolvable problem,'' \emph{Bulletin of
  the American Mathematical Society}, vol.~52, 1946.

\bibitem{sutherland86}
D.~Sutherland, ``A model of information,'' in \emph{Proc. 9th National Computer
  Security Conf.}, 1986, pp. 175--183.

\bibitem{WJ90}
J.~T. Wittbold and D.~M. Johnson, ``Information flow in nondeterministic
  systems.'' in \emph{IEEE Symposium on Security and Privacy}, 1990, pp.
  144--161.

\bibitem{mccullough88}
D.~McCullough, ``Noninterference and the composability of security
  properties,'' in \emph{Proc. IEEE Symp. on Security and Privacy}, 1988, pp.
  177--186.

\bibitem{FG01}
R.~Focardi and R.~Gorrieri, ``Classification of security properties ({Part I:}
  information flow),'' in \emph{Foundations of Security Analysis and Design,
  FOSAD 2000, Bertinoro, Italy, September 2000}, ser. LNCS, R.~Focardi and
  R.~Gorrieri, Eds.\hskip 1em plus 0.5em minus 0.4em\relax Springer, 2001, vol.
  2171, pp. 331--396.

\bibitem{Ryan01}
P.~Ryan, ``Mathematical models of computer security,'' in \emph{Foundations of
  Security Analysis and Design, FOSAD 2000, Bertinoro, Italy, September 2000},
  ser. LNCS, R.~Focardi and R.~Gorrieri, Eds.\hskip 1em plus 0.5em minus
  0.4em\relax Springer, 2001, vol. 2171, pp. 1--62.

\bibitem{BartheDR04}
G.~Barthe, P.~R. D'Argenio, and T.~Rezk, ``Secure information flow by
  self-composition,'' in \emph{CSFW}.\hskip 1em plus 0.5em minus 0.4em\relax
  IEEE Computer Society, 2004, pp. 100--114.

\bibitem{FG96}
R.~Focardi and R.~Gorrieri, ``Automatic compositional verification of some
  security properties,'' in \emph{Proceedings of Second International Workshop
  on Tools and Algorithms for the Construction and Analysis of Systems
  (TACAS'96)}, ser. Springer LNCS, vol. 1055, 1996, pp. 167--186.

\bibitem{Mantel-thesis}
H.~Mantel, ``A uniform framework for the formal speciﬁcation and
  veriﬁcation of information flow security,'' Ph.D. dissertation,
  Universit¨at des Saarlandes, 2003.

\bibitem{DSouzaRS05}
D.~D'Souza, K.~R. Raghavendra, and B.~Sprick, ``An automata based approach for
  verifying information flow properties,'' \emph{Electr. Notes Theor. Comput.
  Sci.}, vol. 135, no.~1, pp. 39--58, 2005.

\bibitem{DSouzaHKRS08}
D.~D'Souza, R.~Holla, J.~Kulkarni, R.~K. Ramesh, and B.~Sprick, ``On the
  decidability of model-checking information flow properties,'' in
  \emph{ICISS}, ser. Lecture Notes in Computer Science, R.~Sekar and A.~K.
  Pujari, Eds., vol. 5352.\hskip 1em plus 0.5em minus 0.4em\relax Springer,
  2008, pp. 26--40.

\bibitem{KB06}
B.~K\"{o}pf and D.~A. Basin, ``Timing-sensitive information flow analysis for
  synchronous systems,'' in \emph{Proc. European Symp. on Research in Computer
  Security}, ser. Springer LNCS, vol. 4189, 2006, pp. 243--262.

\bibitem{CMZ10}
F.~Cassez, R.~van~der Meyden, and C.~Zhang, ``The complexity of synchronous
  notions of information flow security,'' in \emph{FoSSaCS 2010, 13th
  International Conference on Foundations of Software Science and Computation
  Structures}, ser. Springer LNCS, vol. 6014, 2010, pp. 282--296.

\bibitem{FGP95}
R.~G. Riccardo~Focardi and V.~Panini, ``The security checker: a semantics-based
  tool for the verification of security properties,'' in \emph{Proceedings of
  Eighth IEEE Computer Security Foundations Workshop (CSFW'95)}, 1995, pp.
  60--69.

\bibitem{Whalen10}
M.~W. Whalen, D.~A. Greve, and L.~G. Wagner, ``Model checking information
  flow,'' in \emph{Design and Verification of Microprocessor Systems for
  High-Assurance Applications}, D.~Hardin, Ed.\hskip 1em plus 0.5em minus
  0.4em\relax Springer-Verlag, Berlin Germany, March 2010.

\bibitem{GWV03}
D.~Greve, M.~Wilding, and M.~Vanfleet, ``A separation kernel formal security
  policy,'' in \emph{Proc. Fourth International Workshop on the ACL2 Theorem
  Prover and Its Applications}, 2003.

\bibitem{pinsky95}
S.~Pinsky, ``Absorbing covers and intransitive non-interference.'' in
  \emph{Proc. IEEE Symp. on Security and Privacy}, 1995, pp. 102--113.

\bibitem{HLFMY05}
N.~Hadj-Alouane, S.~Lafrance, F.~Lin, J.~Mullins, and M.~Yeddes, ``On the
  verification of intransitive noninterference in multilevel security,''
  \emph{IEEE Trans. on Systems, Man and Cybernetics, Part B}, vol.~35, no.~5,
  pp. 948-- 958, Oct. 2005.

\bibitem{meyden-compini}
R.~van~der Meyden, ``A comparison of semantic models for intransitive
  noninterference,'' Dec 2007, unpublished manuscript, available at {\tt
  http://www.cse.unsw.edu.au/$\sim$meyden}.

\bibitem{ESW13}
S.~Eggert, H.~Schnoor, and T.~Wilke, ``Noninterference with local policies,''
  in \emph{Proceedings of the 38th International Symposium on Mathematical
  Foundations of Computer Science (MFCS 2013)}, June 2013, to appear.

\bibitem{chong-meyden-tark09}
S.~Chong and R.~van~der Meyden, ``Deriving epistemic conclusions from agent
  architecture,'' in \emph{Proc. Conf. on Theoretical Aspects of Knowledge and
  Rationality}, 2009, pp. 61--70.

\bibitem{chong-meyden-law12}
------, ``Using architecture to reason about information security,'' in
  \emph{Layered Assurance Workshop}, 2012, to appear.

\bibitem{Mantel:Sands:APLAS04}
H.~Mantel and D.~Sands, ``Controlled declassification based on intransitive
  noninterference,'' in \emph{Proc. Asian Symp. on Programming Languages and
  Systems}, ser. LNCS, vol. 3302.\hskip 1em plus 0.5em minus 0.4em\relax
  Springer-Verlag, Nov. 2004, pp. 129--145.

\bibitem{Sabelfeld:Sands:CSFW05}
A.~Sabelfeld and D.~Sands, ``Dimensions and principles of declassification,''
  in \emph{Proceedings of the 18th IEEE Computer Security Foundations
  Workshop}.\hskip 1em plus 0.5em minus 0.4em\relax IEEE Computer Society
  Press, 2005, pp. 255--269.

\bibitem{CM04}
S.~Chong and A.~C. Myers, ``Security policies for downgrading,'' in \emph{11th
  ACM Conf. on Computer and Communications Security (CCS)}, Oct 2004.

\bibitem{BanerjeeNR08}
A.~Banerjee, D.~A. Naumann, and S.~Rosenberg, ``Expressive declassification
  policies and modular static enforcement,'' in \emph{IEEE Symposium on
  Security and Privacy}.\hskip 1em plus 0.5em minus 0.4em\relax IEEE Computer
  Society, 2008, pp. 339--353.

\bibitem{MyersSZ04}
A.~C. Myers, A.~Sabelfeld, and S.~Zdancewic, ``Enforcing robust
  declassification,'' in \emph{CSFW}.\hskip 1em plus 0.5em minus 0.4em\relax
  IEEE Computer Society, 2004, pp. 172--186.

\bibitem{young94}
W.~Bevier and W.~Young, ``A state-based approach to noninterference,'' in
  \emph{Proc. IEEE Computer Security Foundations Workshop}, 1994, pp. 11--21.

\bibitem{Mantel01}
H.~Mantel, ``Information flow control and applications - bridging a gap,'' in
  \emph{FME}, ser. Lecture Notes in Computer Science, J.~N. Oliveira and
  P.~Zave, Eds., vol. 2021.\hskip 1em plus 0.5em minus 0.4em\relax Springer,
  2001, pp. 153--172.

\bibitem{oheimb04}
D.~von Oheimb, ``Information flow control revisited: {Noninfluence =
  Noninterference + Nonleakage},'' in \emph{Computer Security -- ESORICS 2004},
  ser. LNCS, vol. 3193, 2004, pp. 225--243.

\bibitem{EMZ12}
K.~Engelhardt, R.~van~der Meyden, and C.~Zhang, ``Intransitive noninterference
  in nondeterministic systems,'' in \emph{Proc. ACM Conf. on Computer and
  Communications Security}, 2012, pp. 869--880.

\bibitem{DBLP:conf/csfw/2004}
\emph{17th IEEE Computer Security Foundations Workshop, (CSFW-17 2004), 28-30
  June 2004, Pacific Grove, CA, USA}.\hskip 1em plus 0.5em minus 0.4em\relax
  IEEE Computer Society, 2004.

\end{thebibliography}

\newpage 

\section*{Appendix: Reduction from TO-security to ITO-security} 

We show that TO-security can be reduced to \itosecty. 

First, we need some new definitions: 
For all $u \in \Dom$ and all $\alpha \in \Actions$, $\lpre_u(\alpha)$
is the largest prefix of $\alpha$ that ends in an action $a$
with $\dom(a) = u$. 
(If there is no such action, then we take $\lpre_u(\alpha)$ to be $\epsilon$.) 
Similar to the definition of $\tview$, we define 
\begin{equation*}
  \ftview_u(\alpha) = \view_u(\lpre_u(\alpha)) \enspace. 
\end{equation*}

Now we can give a characterization of \itosecty, similar to the
characterization of \tosecty\  in 
Proposition~\ref{prop:ptflat}.

\begin{proposition} \label{prop:charn-ito}
  $M$ is \itosec\ with respect to a policy $\nintrel$ iff for all
  sequences $\alpha, \beta \in \Actions^*$, and domains $u \in \Dom$, 
  if 
  $\purge_u(\alpha) = \purge_u(\beta)$ and
  $\ftview_v(\alpha) = \ftview_v(\beta)$ for all domains $v \neq u$ 
  such that $v \nintrel u$, 
  then
  $\obs_u(s_0 \cdot \alpha) = \obs_u(s_0 \cdot \beta)$. 
\end{proposition}

\begin{proof} 
We first argue from left to right. 
We show by induction on the combined length of 
$\alpha, \beta \in \Actions^*$ that, for $u\in \Dom$, if  we have 
$\purge_u(\alpha) = \purge_u(\beta)$ and 
$ \ftview_{v}(\alpha ) = \ftview_{v}(\beta)$ for all domains $v\neq u$ with $v\nintrel u$, 
then $\ito_u(\alpha) = \ito_u(\beta)$. 
Thus, if $M$ is \itosec\ and the antecedent of the right hand side holds, then so does 
the consequent $\obs_u(s_0 \cdot \alpha) = \obs_u(s_0 \cdot \beta)$,  by \itosecty. 
The proof of the induction is trivial in the base case of $\alpha = \beta = \epsilon$. 
Suppose that $\alpha =\alpha'a$ where $a\in \Actions$, and   
$\purge_u(\alpha' a ) = \purge_u(\beta)$ and 
$ \ftview_{v}(\alpha' a ) = \ftview_{v}(\beta)$ for all domains $v\neq u$ with $v\nintrel u$. 
We consider several cases. 

{\em Case 1:} $\dom(a) \not \nintrel u$. Then 
$\purge_u(\alpha' ) = \purge_u(\alpha' a) = \purge_u(\beta)$, and 
for $v\nintrel u$ with $v\neq u$  have $v \neq \dom(a)$, so 
$\ftview_u(\alpha') = \ftview_u(\alpha' a) = \ftview_u(\beta)$. 
By the induction hypothesis, $\ito_u(\alpha') = \ito_u(\beta)$. 
Since $\ito_u(\alpha' a) = \ito_u(\alpha)$ in this case, we obtain 
$\ito_u(\alpha' a) = \ito_u(\beta)$.  

{\em Case 2:} $\dom(a) \nintrel u$. In this case, 
$\purge_u(\alpha') a = \purge_u(\alpha' a) = \purge_u(\beta)$, 
so $\beta = \beta' b$ for some $b\in \Actions$. 
In case $\dom(b) \not \nintrel u$, we may apply Case 1 with the roles 
of $\alpha$ and $\beta$ swapped. We therefore assume without loss of generality that 
$\dom(b) \nintrel u$. Then $\purge_u(\beta) = \purge_u(\beta')b$, and it follows
that $a=b$ and $\purge_u(\alpha') = \purge_u(\beta')$. 

Consider now $v\neq u$ with $v\nintrel u$.  In the case $v\neq \dom(a)$, we have 
$\ftview_v(\alpha'  ) = \ftview_v(\alpha' a)  = \ftview_v(\beta' b)  =  \ftview_v(\beta')$. 
In the case that $v= \dom(a)$, we have 
$\ftview_v(\alpha' a) = \view_v(\alpha' a)$ and $\ftview_v(\beta' b) = \view_v(\beta' b)$, 
so we obtain that  $\view_v(\alpha' a) =  \view_v(\beta' b)$, so also
$\view_v(\alpha' ) =  \view_v(\beta' )$. The latter implies that $\ftview_v(\alpha' ) =  \ftview_v(\beta' )$
in this case also. 

It therefore follows by the induction hypothesis that $\ito_u(\alpha') = \ito_u(\beta')$. 
Thus, in case $\dom(a) \neq u$, using what was shown above, we have 
\begin{align*} 
\ito_u(\alpha'a)  & = (\ito_u(\alpha'), \view_{\dom(a)}(\alpha' a), a) \\
& = (\ito_u(\beta'), \view_{\dom(a)}(\beta' b), b) \\
& = \ito_u(\beta'b)\mathpunct.
\end{align*} 
Similarly, in the case $\dom(a) = u$, we have 
\begin{align*} 
\ito_u(\alpha'a)  & = (\ito_u(\alpha'), \view_{\dom(a)}(\alpha' ), a) \\
& = (\ito_u(\beta'), \view_{\dom(a)}(\beta' ), b) \\
& = \ito_u(\beta'b)\mathpunct.
\end{align*}
Thus, in either case we have $\ito_u(\alpha) = \ito_u(\beta)$, as claimed.  
This completes the proof of the direction from left to right. 

For the converse, define for each domain $u$, a function $P_u$ and functions $F_v$, all 
on the range of the function $\ito_u$, 
by $P_u(\epsilon) =\epsilon$ and $F_v(\epsilon) = \obs_v(s_0)$, and 
$P_u((x,y,a)) =  P_u(x) a$ and 
$$F_v((x,y,a)) =  \begin{cases} 
y & \text{if $\dom(a) = v$,} \\ 
F_v(x) & \text{if $\dom(a)\neq  v$.} 
\end{cases} $$
We claim that for all $u,v \in \Dom$ with $u\neq v$ and $v\nintrel u$, 
and all $\alpha \in \Actions^*$, we have $P_u(\ito_u(\alpha)) = \purge_u(\alpha)$ and 
$F_v(\ito_u(\alpha)) = \ftview_v(\alpha)$. 
Now, if  $M$ is not \itosec\ then  there exist $\alpha, \beta \in \Actions^*$ and domain $u$
with $\ito_u(\alpha) = \ito_u(\beta)$ and $\obs_u(s_0\cdot \alpha) \neq \obs_u(s_0\cdot \beta)$. 
It follows from the claim that $\purge_u(\alpha) = \purge_u(\beta)$ and 
$\ftview_v(\alpha) = \ftview_v(\beta)$ for domains $v\neq u$ with $v\nintrel u$,
so the right hand side of the proposition is also false. 

The proof of the claim that $P_u(\ito_u(\alpha)) = \purge_u(\alpha)$ is a straightforward
induction, left to the reader. For $F_v(\ito_u(\alpha)) = \ftview_v(\alpha)$ we argue inductively, 
as follows. The base case of $\alpha = \epsilon$ is trivial. For $\alpha = \alpha' a $, 
we consider the following cases: 

{\em Case 1: } $\dom(a) \not \nintrel u$. In this case we have $v \neq \dom(a)$. 
Thus, 
\begin{align*} 
F_v(\ito_u(\alpha' a)) & = F_v(\ito_u(\alpha')) \\ 
& = \ftview_v(\alpha') & \text{by induction,} \\ 
& = \ftview_v(\alpha' a) 
\end{align*} 

{\em Case 2:} $\dom(a) \nintrel u$ and $\dom(a) \neq u$ and $v = \dom(a)$. 
In this case 
\begin{align*} 
F_v(\ito_u(\alpha' a)) & = F_v((\ito_u(\alpha'), \view_{\dom(a)}(\alpha' a), a)) \\ 
& = \view_v(\alpha' a)  \\ 
& = \ftview_v(\alpha' a) 
\end{align*} 

{\em Case 3:} $\dom(a) \nintrel u$ and $\dom(a) \neq u$ and $v \neq  \dom(a)$. 
In this case 
\begin{align*} 
F_v(\ito_u(\alpha' a)) & = F_v((\ito_u(\alpha'), \view_{\dom(a)}(\alpha' a), a)) \\ 
& = F_v(\ito_u(\alpha' ))  \\ 
& = \ftview_v(\alpha' ) & \text{by induction}\\ 
& = \ftview_v(\alpha'  a)  
\end{align*} 

{\em Case 4:}  $\dom(a) = u$. 
In this case, 
\begin{align*} 
F_v(\ito_u(\alpha' a)) & = F_v((\ito_u(\alpha'), \view_{u}(\alpha' ), a)) \\ 
& = F_v(\ito_u(\alpha' ))  & \text{since  $v\neq \dom(a)$}\\ 
& = \ftview_v(\alpha' ) & \text{by induction}\\ 
& = \ftview_v(\alpha'  a)  & \text{since  $v\neq \dom(a)$}
\end{align*}
This completes the proof of the claim.  
\end{proof}

We now show to reduce \tosecty\  to \itosecty. Given a system $M$, we construct a 
system  $M'$ such that $M$ is \tosec\ iff $M'$ is \itosec. The intuition for the 
construction is that \itosecty\ permits a faster flow of information than 
\tosecty. In particular, when  action $a$, with $\dom(a) \nintrel u$ and $\dom(a) \neq u$, 
is performed after sequence $\alpha$, 
\itosecty\ permits the  transmission of the information in the view $\view_{\dom(a)}(\alpha a)$, 
whereas \tosecty\ permits transmission only of the  information in the shorter view 
$\view_{\dom(a)}(\alpha )$. The reduction handles this by  replacing the final observation 
$\obs_u(s_0\cdot\alpha a)$ made by $\dom(a)$  by the uninformative observation $\bot$. 
This makes the faster flow of \itosecty, in the system $M'$ equivalent to the slower
flow of \tosecty\ in $M$. 

More precisely, given a system $M = \langle \States, s_0, \Actions, \step, \obs, \dom
\rangle$, 
we define a new system $M' = \langle \States', s_0', \Actions',
\step', \obs', \dom' \rangle$ as follows. 
We define the set of final actions $\final{\Actions} = \{ \final{a} |
a \in \Actions \}$. 
For every such an action we have $\dom(\final{a}) = \dom(a)$. 
The idea 
underlying
final actions is that, if an agent has performed one of
its final actions, all further actions of this agent are ignored by
the system. 
Therefore the system has to keep track of which agent has performed one
of its final actions. 
If an agent has performed one of its final actions, 
its observation
is set to $\bot$ and will never change. 
More 
formally:  
\begin{align*}
  \States'& = \States \times \powerset{\Dom} \\
  s_0' & = (s_0, \emptyset) \\
  \Actions' & = \Actions \cup \final{\Actions} \\
\intertext{For any $s \in \States$ and $U \subseteq \Dom$:}
  \obs_u' (s, U) & = 
  \begin{cases}
    \obs_u(s) & \text{if } u \not\in U \\
    \bot & \text{otherwise} 
  \end{cases}
\\
\intertext{For any $s\in \States$, $a \in \Actions$
  and $U \subseteq
  \Dom$:}
(s, U) \cdot a & = 
\begin{cases}
  (s\cdot a, U) & \text{ if $\dom(a) \not\in U$ and $a \in
    \Actions$}\\
  (s \cdot b, U \cup \{\dom(a)\}) & 
  \text{ if $\dom(a) \not \in U$ and $a \in \final{\Actions}$ with $a
    = \final{b}$} \\
  (s, U) & \text{ if $\dom(a) \in U$}
\end{cases}
\end{align*}
For the functions $\view$, $\tview$, etc., we use the same notation in
both systems. 
The intended system will be clear from the context.

\begin{lemma}
  A system $M$ is \tosec\ with respect to a policy $\nintrel$ iff
  the system $M'$ is \itosec\ with respect to the same policy
  $\nintrel$. 
\end{lemma}
\begin{proof}
We first show the implication from left to right. 
 We begin by defining a function that transfers between runs in the two systems, and 
 prove several of its properties. 
  \newcommand{\down}[1]{\overline{#1}}
  For an action $a \in A'$, we define $\down{a} = a$ if $a\in A$ and 
 $\down{a}   = b $ if $a = \final{b}\in \final{\Actions}$. 
  We define the function 
  $\convertback \colon \Actions'^* \rightarrow \Actions^*$ 
  by  
  $\convertback(\epsilon) = \epsilon$ and
$$\convertback(\alpha a) = 
  \begin{cases}
    \convertback(\alpha) & \text{ if there is a final action of
      $\dom(a)$ in $\alpha$} \\
        \convertback(\alpha)\,\down{a} & \text{ if 
        there is no final action of $\dom(a)$ in $\alpha$.} 
  \end{cases}
$$ 
This function $\convertback$ removes for each agent its actions
performed after the agent's first final action. This first final action
is 
replaced 
by the corresponding action from $\Actions$. 

We observe, that we have for any $\alpha \in \Actions'^*$: 
If $s_0' \cdot \alpha = (s, U)$ for some $s \in \States$ and $U
\subseteq \Dom$, then $s= s_0 \cdot \convertback(\alpha)$. 
This implies that for any $u \in \Dom$ and any $\alpha \in \Actions'^*$ that does not
contain a final action of $u$, we have
\begin{equation*}
  \obs_u(s_0 \cdot \convertback(\alpha)) = \obs
'_u(s_0'\cdot \alpha) \enspace. 
\end{equation*}
This also implies that for any $u \in \Dom$ and any $\alpha\in \Actions'^*$ that does 
not contain any final action of $u$, we have
\begin{equation*}
  \view_u(\convertback(\alpha)) =  \view_u(\alpha) \enspace. 
\end{equation*}
(Here the left hand side is computed in $M$ and the right hand side in $M'$.) 
The argument for this is an induction on $\alpha$. The claim is trivial in case $\alpha = \epsilon$. 
Suppose $\alpha = \alpha' a $ where $a\in \Actions$, 
and $\alpha' a$ does not contain any final actions of $u$. We consider
several cases: 
\begin{itemize} 
\item {\em Case 1:} $\dom(a) = u$. Then $a$ is not final and $\alpha'$ contains no final action of $u$,  
so $\convertback(\alpha' a ) = \convertback(\alpha') a$, and 
\begin{align*} 
\view_u(\alpha' a) & = \view_u(\alpha') \, a \, \obs'_u(s'_0\cdot \alpha' a) \\ 
& = \view_u(\convertback(\alpha')) \, a\, \obs'_u(s'_0\cdot \alpha' a)  & \text{by induction}\\ 
& =  \view_u(\convertback(\alpha'))\, a \,\obs_u(s_0\cdot \convertback(\alpha' a)) \\ 
& = \view_u(\convertback(\alpha')) \,a \,\obs_u(s_0\cdot \convertback(\alpha') a) \\ 
& = \view_u(\convertback(\alpha')a) \\ 
& = \view_u(\convertback(\alpha' a)) \enspace.
\end{align*}  
\item {\em Case 2:} $\dom(a) \neq u$ and there is no final action of $\dom(a)$ in $\alpha'$. 
Then $\convertback(\alpha' a ) = \convertback(\alpha') \down{a}$, so 
 \begin{align*} 
\view_u(\alpha' a) & = \view_u(\alpha') \circ \obs'_u(s'_0\cdot \alpha' a) \\ 
& = \view_u(\convertback(\alpha')) \circ \obs'_u(s'_0\cdot \alpha' a)  & \text{by induction}\\ 
& =  \view_u(\convertback(\alpha'))\circ \obs_u(s_0\cdot \convertback(\alpha' a)) \\ 
& = \view_u(\convertback(\alpha')) \circ\obs_u(s_0\cdot \convertback(\alpha') a) \\ 
& = \view_u(\convertback(\alpha')a) \\ 
& = \view_u(\convertback(\alpha' a)) \enspace.
\end{align*}  

\item {\em Case 3:} $\dom(a) \neq u$ and there is a final action of
  $\dom(a)$ in $\alpha'$. 
Then $\convertback(\alpha' a ) = \convertback(\alpha') $, so 
 \begin{align*} 
\view_u(\alpha' a) & = \view_u(\alpha') \circ \obs'_u(s'_0\cdot \alpha' a) \\ 
& = \view_u(\convertback(\alpha')) \circ \obs'_u(s'_0\cdot \alpha' a)  & \text{by induction}\\ 
& =  \view_u(\convertback(\alpha'))\circ \obs_u(s_0\cdot \convertback(\alpha' a)) \\ 
& = \view_u(\convertback(\alpha')) \circ\obs_u(s_0\cdot \convertback(\alpha')) \\ 
& = \view_u(\convertback(\alpha')) \\ 
& = \view_u(\convertback(\alpha' a)) \enspace.
\end{align*}

\end{itemize} 
This completes the proof that $\view_u(\convertback(\alpha)) =  \view_u(\alpha)$. 
Plainly, this implies that for  any $\alpha,\beta\in \Actions'^*$ that do 
not contain any final action of $u$, we have
\begin{equation*}
  \view_u(\convertback(\alpha)) = \view_u(\convertback(\beta)) 
\text{ iff } 
\view_u(\alpha) = \view_u(\beta) \enspace. 
\end{equation*}

We now claim that for any $u \in \Dom$ and any 
$\alpha, \beta\in \Actions'^*$ that do not contain a final action of domain $u$,  
that 
if $\purge_u(\alpha) = \purge_u(\beta)$ 
and 
$\ftview_v(\alpha) = \ftview_v(\beta)$ for all $v \neq u$, $v
\nintrel u$, then
$\purge_u(\convertback(\alpha)) = \purge_u(\convertback(\beta))$
and 
$\tview_v(\convertback(\alpha)) = \tview_v(\convertback(\beta))$ for
all $v \neq u$, $v \nintrel u$. 

We prove the claim by induction over the combined length of $\alpha$ and 
$\beta$.
It is clear that it holds for $\alpha = \beta = \epsilon$. 
Consider $\alpha = \alpha' a$ and $\beta$,  
neither containing a final action of domain $u$,  
such that
$\purge_u(\alpha' a) = \purge_u(\beta)$ and 
for all $v \neq u$, $v \nintrel u$: $\ftview_v(\alpha' a) =
\ftview_v(\beta)$ and that the implication holds for strings of
shorter combined length. 
We consider several cases: 
\begin{itemize} 
\item {\em Case 1:}  $\dom(a) \not\nintrel u$.  
This gives
\begin{equation*}
  \purge_u(\alpha') = \purge_u(\alpha' a) = \purge_u(\beta)
\end{equation*}
and 
\begin{equation*}
  \ftview_v(\alpha') = \ftview_v(\alpha' a) = \ftview_v(\beta)
  \end{equation*}
for all $v \neq u$, $v\nintrel u$.
By the induction hypothesis, 
\begin{equation*}
  \purge_u(\convertback(\alpha')) = \purge_u(\convertback(\beta))
\end{equation*}
and 
\begin{equation*}
  \tview_v(\convertback(\alpha')) =  \tview_v(\convertback(\beta))
  \end{equation*}
for all $v \neq u$, $v\nintrel u$.
If there is a final action of $\dom(a)$ in $\alpha'$, then 
$\convertback(\alpha' a ) = \convertback(\alpha')$ 
and the claim for the pair $\alpha' a$ and $\beta$ is immediate. 
We assume in the following that
there is no final action of $\dom(a)$ in $\alpha'$. 
Then, 
since $\dom(\down{a}) = \dom(a) \not \nintrel u$, 
 we have
\begin{align*}
  \purge_u(\convertback(\alpha' a)) & = \purge_u(\convertback(\alpha')
  \down{a}) \\
  & = \purge_u(\convertback(\alpha')) \\
  & = \purge_u(\convertback(\beta )) 
\end{align*}
and for all $v \neq u$, $v \nintrel u$: 
\begin{align*}
  \tview_v(\convertback(\alpha' a)) & = \tview_v(\convertback(\alpha')
 \down{a}) \\
  & = \tview_v(\convertback(\alpha'))  \\
  & = \tview_v(\convertback(\beta )) \enspace. 
\end{align*}

\item {Case 2: } 
$\dom(a) \nintrel u$. 
In this case we have 
$\purge_u(\alpha' ) a = \purge_u(\alpha'  a) = \purge_u(\beta)$, 
so $\beta = \beta' b$ for some action $b$. If $\dom(b) \not \nintrel u$, then 
we may swap the roles of $\alpha$ and $\beta$ and apply the previous case. 
Hence, without loss of generality, $\dom(b) \nintrel u$, 
and we have  $\purge_u(\beta) = \purge_u(\beta') b$. 
It follows that $a=b $ and $\purge_u(\alpha') = \purge_u(\beta')$. 
For all $v \neq u$, $v \nintrel u$ with $v \neq \dom(a)$ we have
\begin{align*}
  \ftview_v(\alpha') & = \ftview_v(\alpha' a) \\
  & = \ftview_v(\beta' a) \\
  & = \ftview_v(\beta') \enspace. 
\end{align*}
If $v = \dom(a)$ then we have
\begin{align*}
  \view_v(\alpha' a) & = \ftview_v(\alpha' a) \\
  & = \ftview_v(\beta' a) \\
  & = \view_v(\beta' a) \enspace. 
\end{align*}
This implies $\view_v(\alpha') = \view_v(\beta')$ and therefore,
$\ftview_v(\alpha') = \ftview_v(\beta')$. 

Thus, by the induction hypothesis, 
$$\purge_u(\convertback(\alpha')) = \purge_u(\convertback(\beta'))$$ and 
$\tview_v(\convertback(\alpha')) =  \tview_v(\convertback(\beta'))$ 
for all $v\neq u$, $v\nintrel u$. 

If there is a final action of $\dom(a)$ in $\alpha'$ then, by $\purge_u(\alpha') = \purge_u(\beta')$, 
there is also a final action of $\dom(a)$ in $\beta'$, and we have 
$\convertback(\alpha' a) = \convertback(\alpha')$   and $\convertback(\beta' a) = \convertback(\beta')$. 
The desired conclusion is then direct from the above inductive conclusion. Alternately, 
if there is no final action of $\dom(a)$ in $\alpha'$ then there is no final action of $\dom(a)$ in $\beta'$. 
In this case, since $\dom(\down{a})= \dom(a) \nintrel u$, we have 
\begin{align*}
  \purge_u(\convertback(\alpha' a)) & = \purge_u(\convertback(\alpha') \down{a})
  \\
  & = \purge_u(\convertback(\alpha'))\down{a} \\
    & = \purge_u(\convertback(\beta'))\down{a} \\
    & = \purge_u(\convertback(\beta') \down{a}) \\
  & = \purge_u(\convertback(\beta'a)) \enspace.
\end{align*}
Also, for all $v \neq u$, $v \nintrel u$,  
in case $\dom(a) \neq v$ we have 
\begin{align*}
  \tview_v(\convertback(\alpha' a)) & = \tview_v(\convertback(\alpha')
  \down{a}) \\
  & = \tview_v(\convertback(\alpha'))\\
  & = \tview_v(\convertback(\beta'))\\
  & = \tview_v(\convertback(\beta') \down{a}) \\
  & = \tview_v(\convertback(\beta' a)) \enspace.   
\end{align*}
In the case $\dom(a) = v$ we argue as follows. 
If $a$ is not a final action, 
we have
\begin{align*}
  \view_{\dom(a)}(\alpha' a) & = \ftview_{\dom(a)}(\alpha' a) \\
  & = \ftview_{\dom(a)}(\beta' a) \\
  & = \view_{\dom(a)}(\beta' a) \enskip .
\end{align*}
Since there is no final action of $\dom(a)$ in $\alpha'$ or $\beta'$, and $a$ is not final, 
there is no final action of $\dom(a)$ in $\alpha' a$ or $\beta' a$,
and therefore
$$\view_{\dom(a)}(\convertback(\alpha' a)) =
\view_{\dom(a)}(\convertback(\beta' a))\enskip,$$ 
using the equivalence proved above. From this equation it follows that 
$$\tview_{\dom(a)}(\convertback(\alpha' a)) =
\tview_{\dom(a)}(\convertback(\beta' a))\enskip.$$ 

In the case that $a$ is a final action, 
we have
\begin{align*}
\view_{\dom(a)}(\alpha') a \bot 
& = 
  \ftview_{\dom(a)}(\alpha' a) \\
  & = \ftview_{\dom(a)}(\beta' a) \\
  & = \view_{\dom(a)}(\beta') a \bot \enspace. 
\end{align*}
Therefore $\view_{\dom(a)}(\alpha') = \view_{\dom(a)}(\beta')$. 
Since there is no final action of $\dom(a)$ in $\alpha'$ or $\beta'$,  it follows using the equivalence proved above that 
$$\view_{\dom(a)}(\convertback(\alpha')) = \view_{\dom(a)}(\convertback(\beta'))\enskip . $$ 
Thus,
\begin{align*}
  \tview_{\dom(a)}(\convertback(\alpha' a)) & = \tview_{\dom(a)}(\convertback(\alpha')\,
  \down{a}) \\
  & = \view_{\dom(a)}(\convertback(\alpha'))\, \down{a} \\
  & = \view_{\dom(a)}(\convertback(\beta'))\, \down{a} \\
  & = \tview_{\dom(a)}(\convertback(\beta' a)) 
  \enspace. 
\end{align*}
\end{itemize} 
This completes the proof of the claim. 

We are now positioned to prove that if $M$ is \tosec\ then $M'$ is \itosec.  
We show the contrapositive. Suppose $M'$ is not \itosec. 
By Proposition~\ref{prop:charn-ito}, there exist $\alpha, \beta \in \Actions'^*$ and domain $u$ 
such that $\purge_u(\alpha) = \purge_u(\beta)$ and
  $\ftview_v(\alpha) = \ftview_v(\beta)$ for all domains $v \neq u$ 
  such that $v \nintrel u$, and  $\obs'_u(s'_0 \cdot \alpha) \neq \obs'_u(s'_0 \cdot \beta)$. 
It follows from $\purge_u(\alpha) = \purge_u(\beta)$ that 
$\alpha$ contains a final action of domain $u$ iff $\beta$ contains a final action of  domain $u$. 
But if both contain such a  final  action, then $\obs'_u(s'_0 \cdot \alpha) = \bot =  \obs'_u(s'_0 \cdot \beta)$, 
contrary to assumption. Thus neither $\alpha$ nor $\beta$ contain a  final action of $u$. 
By what was shown above, $\purge_u(\convertback(\alpha)) = \purge_u(\convertback(\beta))$ and
  $\ftview_v(\convertback(\alpha)) = \ftview_v(\convertback(\beta))$ for all domains $v \neq u$ 
  such that $v \nintrel u$. 
  Also, $\obs_u(s_0 \cdot \convertback(\alpha)) \neq  \obs_u(s_0 \cdot \convertback(\beta))$.
  By the characterization of  \tosecty\  of Proposition~\ref{prop:ptflat}, this 
implies that  $M$ is not \tosec\ with respect to $\nintrel$.

For the other direction of the proof we define a function
$\convert \colon \Dom \times \Actions^* \rightarrow \Actions^*$, 
which, for each domain $v\neq u$ with $v \nintrel u$, replaces 
the rightmost action $a$ with $\dom(a) = v$ by the action  $\final{a}$.

We observe that,  for all $u \in \Dom$, if $\gamma, \gamma'$ are prefixes of 
$\alpha$ and $\convert_u(\alpha)$, respectively, and have the same length, then 
for all $U \subseteq \Dom$, 
if $s_0'\cdot \gamma' = (s, U)$ then $s = s_0 \cdot \gamma$. 
Therefore, since $\convert_u(\alpha)$ contains no final action of domain $u$, 
we have $\obs_u(s_0\cdot \convert_u(\alpha)) = \obs_u(s_0\cdot \alpha)$.  
Moreover,  if $\gamma,\gamma'$ are prefixes of $\alpha$ and $\convert_u(\alpha)$, respectively, 
have the same length, and $\gamma$ does not contain the rightmost action of domain $v$ in $\alpha$ (if any), 
then $\view_v(\gamma) = \view_v(\gamma')$. 

We show that, for all $u \in \Dom$ and all $\alpha, \alpha' \in
\Actions^*$,  
if $\purge_u(\alpha) = \purge_u(\alpha')$ and 
$\tview_v(\alpha) = \tview_v(\alpha')$ for all $v \neq u$, $v
\nintrel u$, 
then 
$\purge_u(\convert_u(\alpha)) = \purge_u(\convert_u(\alpha'))$
and 
$\ftview_v(\convert_u(\alpha)) = \ftview_v(\convert_u(\alpha'))$ for
all $v \neq u$, $v \nintrel u$. 
By an argument similar to that for the opposite direction, it then follows that 
if $M'$ is \itosec\ then $M$ is \tosec. 

We observe that, for any agent $u$, the functions $\purge_u$ and
$\convert_u$ commute. 
This shows that for all $\alpha, \alpha' \in \Actions^*$,  
if $\purge_u(\alpha) = \purge_u(\alpha')$ then, 
$\purge_u(\convert_u(\alpha)) = \purge_u(\convert_u(\alpha'))$. 

To complete the argument, we assume 
$\purge_u(\alpha) = \purge_u(\alpha')$ and 
for a domain $v\neq u$ with $v\nintrel u$, we have
$\tview_v(\alpha) = \tview_v(\alpha')$, and show that  $\ftview_v(\convert_u(\alpha)) = \ftview_v(\convert_u(\alpha'))$. From $\purge_u(\alpha) = \purge_u(\alpha')$ it follows 
that the sequences of actions of domain $v$ in $\alpha$ and $\alpha'$ are the same. 
In particular, if neither sequence contains an action of domain $v$, then the claim is trivial. 
Suppose that $a$ is the last action of domain $v$ in both $\alpha$ and $\alpha'$. 
Then we may write $\alpha = \alpha_1 a \alpha_2$ and $\alpha' = \alpha'_1 a \alpha'_2$, 
where $\alpha_2, \alpha_2'$ contain no actions of domain $v$. Since 
$\tview_v(\alpha) = \tview_v(\alpha')$, we have $\view_v(\alpha_1) = \view_v(\alpha'_1)$. 
Also
$\convert_u(\alpha) = \gamma_1 \final{a} \gamma_2$ and $\convert_u(\alpha') = \gamma'_1 \final{a} \gamma'_2$, 
where $\gamma_1, \gamma'_1$ are of the same length as $\alpha_1, \alpha_1'$, respectively, 
and $\gamma_2, \gamma_2'$ contain no actions of domain $v$. 
Thus, using the observation above, 
\begin{align*}
  \ftview_v(\convert_u(\alpha )) &  = \ftview_v(\gamma_1 \final{a} \gamma_2) \\
  & = \view_v(\gamma_1)\, \final{a} \bot \\
  & = \view_v(\alpha_1) \,\final{a} \bot \\
  & = \view_v(\alpha'_1)\, \final{a} \bot \\
  & = \view_v(\gamma'_1)\, \final{a} \bot \\
  & = \ftview_v(\convert_u(\alpha')) 
  \enspace. 
\end{align*}
\end{proof}

\end{document}